\documentclass[11 pt,a4 paper]{article}
\usepackage[utf8]{inputenc}
\usepackage[OT1]{fontenc}
\usepackage[english]{babel}
\usepackage{dsfont}
\usepackage{amsfonts}
\usepackage{amsmath}
\usepackage{wasysym}
\usepackage{amsthm}
\usepackage{amssymb}
\usepackage{amsthm}
\usepackage{float}
\usepackage{textcomp}
\usepackage{xcolor}
\usepackage{graphicx}
\usepackage{wrapfig}
\usepackage{tikz}
\usepackage{amssymb}
\usepackage{fancybox}
\usepackage{fancyhdr}
\usepackage{mathrsfs}
\usepackage{tikz}
\usepackage[all]{xy}
\usepackage{centernot}
\usepackage{listings}
\usepackage{faktor}
\usepackage{bbm}
\usepackage{bm}
\usepackage{hyperref}
\usepackage{newlfont}
\usepackage{geometry}
\usepackage{ccaption}
\usetikzlibrary{matrix,arrows,decorations.pathmorphing}
\usepackage[backend=biber,style=numeric,sorting=nyt,maxbibnames=9]{biblatex} 
\usepackage{guit}
\usepackage{calligra,amsmath,amssymb}

\theoremstyle{definition}
\newtheorem*{Def}{Definition}
\newtheorem{prop}{Proposition}[section]
\newtheorem*{fact}{Fact}

\newtheorem{obs}[prop]{Remark}
\newtheorem{cor}[prop]{Corollary}

\newtheorem{teo}{Theorem}

\newcommand{\scr}{\mathscr}

\DeclareMathOperator{\Sym}{Sym}

\DeclareMathOperator{\SL}{SL}
\DeclareMathOperator{\GL}{GL}
\DeclareMathOperator{\LG}{LG}
\DeclareMathOperator{\Sp}{Sp}

\DeclareMathOperator{\Gr}{Gr}
\DeclareMathOperator{\Pl}{Pl}
\DeclareMathOperator{\U}{U}
\DeclareMathOperator{\Aut}{Aut}
\DeclareMathOperator{\Mat}{Mat}

\title{Graph States and the Variety of Principal Minors}
\author{Vincenzo Galgano\footnote{Dipartimento di Matematica, Università di Trento, Via Sommarive 14 38123 Trento, Italy; ORCID: 0000-0001-8778-575X (vincenzo.galgano@unitn.it)} , Frédéric Holweck\footnote{Laboratoire Interdisciplinaire Carnot de Bourgogne, ICB/UTBM, UMR6303 CNRS, Universit\'e Bourgogne Franche-Comt\'e, F-90010 Belfort, France; (frederic.holweck@utbm.fr)}}
\date{}

\begin{document}

\maketitle

\begin{abstract}
In Quantum Information theory, graph states are quantum states defined by graphs. In this work we exhibit a correspondence between graph states and the variety of binary symmetric principal minors, in particular their corresponding orbits under the action of $SL(2,\mathbb F_2)^{\times n}\rtimes \mathfrak S_n$. We start by approaching the topic more widely, that is by studying the orbits of maximal abelian subgroups of the $n$-fold Pauli group under the action of $\mathcal C_n^{\text{loc}}\rtimes \mathfrak S_n$, where $\mathcal C_n^{\text{loc}}$ is the $n$-fold local Clifford group: we show that this action corresponds to the natural action of $SL(2,\mathbb F_2)^{\times n}\rtimes \mathfrak S_n$ on the variety $\mathcal Z_n\subset \mathbb P(\mathbb F_2^{2^n})$ of principal minors of binary symmetric $n\times n$ matrices. The crucial step in this correspondence is in translating the action of $SL(2,\mathbb F_2)^{\times n}$ into an action of the local symplectic group $Sp_{2n}^{\text{loc}}(\mathbb F_2)$ on the Lagrangian Grassmannian $LG_{\mathbb F_2}(n,2n)$. We conclude by studying how the former action restricts onto stabilizer groups and stabilizer states, and finally what happens in the case of graph states.\\
\hfill\break
\end{abstract}

\tableofcontents


\section{Introduction}
In Quantum Information, stabilizer states are quantum states known in particular for quantum error correcting codes \cite{gottesman1997}. In the stabilizer formalism, a stabilizer state is described by a maximal $n$-fold abelian subgroup $\mathcal{S}_{M_1,\dots,M_n}$ of the $n$-fold Pauli group $\mathcal{P}_n$ (see Sec.\ref{preliminaries} for definitions) that stabilizes it. Graph states are a special class of quantum stabilizer states elegantly described by a graph $G=(V,E)$ that encodes its stablizer group. Graph states have many applications in quantum information processing \cite{hein2006}: they are in particular  useful for Measured Based Quantum Computation (MBQC) \cite{hein2006,briegel2009,mhalla2011}, for quantum error correcting codes \cite{gottesman1997} and for secret sharing \cite{markham2008,bell2014}.
As a resource for quantum information, it is interesting to propose classification of graph states. A natural framework to classify quantum states is to consider the group of  local unitary operations LU. However for stabilizer states (and graph states), one usually restricts to considering the group of local  unitaries within the Clifford group \cite{article1,van2005}. We will denote by $\mathcal{C}^{\text{loc}} _n\subset \text{LU}$ the group of local Clifford acting on $n$ qubit states. Under the action of $\mathcal{C}^{\text{loc}} _n\rtimes \mathfrak{S}_n$, graph states have been classified up to $n=12$ qubits  \cite{hein2006,cabello2009,cabello2011}.\\
\indent The variety $\mathcal{Z}_n$ of principal minors for $n\times n$ symmetric matrices over a field $\mathbb{K}$ \cite{oeding2011} is an algebraic variety of $\mathbb{P}(\mathbb{K}^{2^n})$ introduced by Holtz and Sturmfels \cite{holtz2007} in order to study relations among principal minors of symmetric matrices. The existence problem of a matrix statisfying predefinite conditions on its principal minors has many applications to matrix theory, probability, statistical physics and computer vision \cite{griffin2006, kenyon2014, oeding2017}.  The goal of this paper is to show another potential application of the study of this variety over the two-elements field $\mathbb{K}=\mathbb{F}_2$ by establishing a bijection between classes of graph states and orbits of the variety $\mathcal{Z}_n$.  More precisely, the main result of this paper is the following theorem.

\begin{teo}\label{theo:main}
 The Lagrangian mapping induces a bijection between the 
$\left(\mathcal{C}^{\text{loc}} _n\rtimes \mathfrak{S}_n\right)$-orbits of maximal abelian subgroups of $\mathcal{P}_n$
 and the  $\left(\SL(2,\mathbb{F}_2)^{\times n}\rtimes \mathfrak{S}_n\right)$-orbits of $\mathcal{Z}_n\subset \mathbb{P}(\mathbb{F}_2^{2^n})$. In particular, there is a one to one correspondence between representatives of the graph states classification, up to
 $\left(\mathcal{C}^{\text{loc}} _n\rtimes \mathfrak{S}_n\right)$-action, and the representatives of the $\left(\SL(2,\mathbb{F}_2)^{\times n}\rtimes \mathfrak{S}_n\right)$-orbits of $\mathcal{Z}_n$.
\end{teo}

 Regarding the cardinality of the orbits, the Lagrangian mapping (Sec.\ref{preliminaries}) shows that, if $\mathcal{O}_i$ is a $\left(\SL(2,\mathbb{F}_2)^{\times n}\rtimes \mathfrak{S}_n\right)$-orbit of $\mathcal{Z}_n$, then the number of corresponding stabilizer states is $4^n|\mathcal{O}_i|$.
 
 Maximal ($n$-fold) abelian subgroups of $\cal P_n$ correspond to subspaces of maximal dimension in the symplectic polar space $\mathcal{W}(2n-1,2)$ of rank $n$ and order $2$ (see Sec.\ref{preliminaries}). The bijection induced by the Lagrangian mapping between subspaces of maximal dimension in $\mathcal{W}(2n-1,2)$ and points of $\mathcal{Z}_n$ was already established in \cite{holweck2014} in order to generalize observations made in \cite{levay2013} regarding the case $n=3$ and its connection with the so-called {\em black-holes/qubits correspondence}. More recently, that same bijection was also considered in \cite{van2019} with motivating examples from supergravity theory. It was proven \cite{van2019} that, over $\mathbb{F}_2$, $\mathcal{Z}_n$ is the image of the Spinor variety and thus a $\text{Spin}(2n+1)$-orbit. However, the correspondence of orbits as established in Theorem \ref{theo:main} was not proven in the former papers, neither was the connection with graph states classification. \\
\indent The paper is organized as follows. In Sec.\ref{preliminaries} we recall the basic definitions regarding the $n$-qubit Pauli group, the symplectic polar space and the Lagrangian mapping. In Sec.\ref{orbitZn} and \ref{orbitS} we show how the $\left(\mathcal{C}^{\text{loc}} _n\rtimes \mathfrak{S}_n\right)$ action on the symplectic polar space translates into an action on the variety of principal minors, proving the first part of Theorem \ref{theo:main}. In Sec.\ref{sec:stab} and \ref{graph} we recall the definitions and basic properties of stabilizer states and graph states, and we complete the proof of our Theorem. Finally Sec.\ref{sec:application} is dedicated to applications of our correspondence.



\section{Preliminaries}\label{preliminaries}
\indent \indent In this section we recall the definitions of the $n$-fold Pauli group and the Clifford group, we introduce the symplectic polar space of rank $n$ and order $2$, encoding the commutation relations of $\mathcal{P}_n$, and finally we describe the Lagrangian mapping.

\subsection{The Pauli group $\cal P_n$ and the local Clifford group $\cal C_n^{\text{loc}}$}

\indent \indent The group of the elementary {\em Pauli matrices} is $\cal P_1=\langle i X,  iZ , i Y \rangle \subset \U(2,\mathbb C)$ where
{\small \[ X= \begin{bmatrix}
0 & 1 \\ 1 & 0
\end{bmatrix} \ \ , \ \ Z=\begin{bmatrix}
1 & 0 \\ 0 & -1
\end{bmatrix} \ \ , \ \ Y=\begin{bmatrix}
0 & -i \\ i & 0
\end{bmatrix} \ .\]}
\noindent We notice that $X^2=Z^2=Y^2=I$ and the following commutation rules hold
\[XZ=-ZX=-iY \ \ , \ \ XY=-YX=iZ \ \ , \ \ ZY=-YZ=-iX \ .\]
Moreover, $\#\cal P_1=16$ and its center is $Z(\cal P_1)=\{\pm I, \pm i I\}\simeq \mathbb Z/4\mathbb Z$: in particular, $V_1=\cal P_1/Z(\cal P_1) \simeq \mathbb Z/2\mathbb Z \times \mathbb Z/2\mathbb Z$.

\begin{Def} The {\bf $n$-fold Pauli group} is $\cal P_n=\{A_1 \otimes \ldots \otimes A_n \ | \ A_i \in \cal P_1\} \subset \U(2^n,\mathbb C)$. 
\end{Def}
	
\noindent From the commutation rules above it is clear that
\[ \cal P_n=\big\{\pm Z^{\mu_1}X^{\nu_1}\otimes \ldots \otimes Z^{\mu_n}X^{\nu_n}, \pm i Z^{\mu_1}X^{\nu_1}\otimes \ldots \otimes Z^{\mu_n}X^{\nu_n} \ \big| \ \mu_i,\nu_i \in \{0,1\}\big\} \ .\]
In particular, we can exhibit the following generators
\begin{equation}\label{generators of P_n}
	\cal P_n=\left\langle I\otimes \ldots \otimes \underbrace{Z}_{l-th} \otimes \ldots \otimes I , \ I\otimes \ldots \otimes \underbrace{X}_{s-th} \otimes \ldots \otimes I \ \bigg| \ 1\leq l,s \leq n \right\rangle \ .
	\end{equation}
Notice that $\#\cal P_n=4\cdot 4^n$ and its center is $Z(\cal P_n)=\{\pm I^{\otimes n}, \pm i I^{\otimes n}\}\simeq \mathbb Z/4\mathbb Z$.

\begin{obs}\label{coordinates}
	The quotient $V_n=\cal P_n/Z(\cal P_n)$ is in one-to-one correspondence with $\mathbb F_2^{2n}$
	\begin{equation}\label{lessoncoordinates} \begin{matrix}
	V_n & \stackrel{1:1}{\longleftrightarrow} & \mathbb F_2^{2n}\\
	[Z^{\mu_1}X^{\nu_1}\otimes \ldots \otimes Z^{\mu_n}X^{\nu_n}] & \longleftrightarrow & (\mu_1, \nu_1, \ldots , \mu_n, \nu_n)
	\end{matrix} \ .\end{equation}
	Clearly, this correspondence is not unique, but it depends on the coordinates we choose in $\mathbb F_2^{2n}$: for instance, another one-to-one correspondence is given by
	\begin{equation}\label{articlecoordinates} [Z^{\mu_1}X^{\nu_1}\otimes \ldots \otimes Z^{\mu_n}X^{\nu_n}] \longleftrightarrow (\mu_1, \mu_2 , \ldots , \mu_n, \nu_1, \nu_2, \ldots , \nu_n) \ .\end{equation}
	\end{obs}

\begin{Def} The {\bf $n$-fold Clifford group} is the normalizer 
	\[ \cal C_n= N_{\U(2^n,\mathbb C)}(\cal P_n)=\big\{U \in \U(2^n,\mathbb C) \ | \ U\cal P_n U^\dagger=\cal P_n\big\} \]
	where $U^\dagger = \ \!\! ^t\overline{U}$ is the hermitian (i.e. conjugated transposed) of $U$.
	\end{Def}

\noindent It is known that $\cal C_n=\langle H_j, \sqrt{Z}_k, CNOT_{st}\rangle$ where
{\small \[ H= \frac{1}{\sqrt{2}}\begin{bmatrix}
1 & 1 \\ 1 & -1
\end{bmatrix} \ \ , \ \ \sqrt{Z}=\begin{bmatrix}
1 & 0 \\ 0 & i
\end{bmatrix} \ \ , \ \ CNOT= \begin{bmatrix}
I_2 & 0 \\ 0 & X
\end{bmatrix}\]}
and $H_j=I \otimes \ldots \otimes \overbrace{H}^{\text{$j$-th}} \otimes \ldots \otimes I$ (same for $\sqrt{Z}_k$), while $CNOT_{st}$ acts as $CNOT$ on the $s$-th and $t$-th factors of $\mathcal{P}_n\subset \U(2,\mathbb C)^{\otimes n}$ and as the identity on the remaining ones. More precisely, there is the following action of $CNOT_{st}$ on the generators of $\mathcal{P}_n$:
{\small \[\begin{array}{l|l}
P\in \mathcal{P}_2 & (CNOT) P (CNOT)^\dagger\\
\hline
X\otimes I & X\otimes X\\
I\otimes X & I\otimes X\\
Z\otimes I & Z\otimes I\\
I\otimes Z & Z\otimes Z
\end{array}\]}

\indent Since not all matrices in $\cal C_n$ are decomposable (i.e. of the form $U_1\otimes \ldots \otimes U_n$ with $U_i \in \U(2,\mathbb C)$), it makes sense to give the next definition.

\begin{Def}
The {\bf $n$-fold local Clifford group} is the subgroup
\[ \cal C_n^{\text{loc}}=\{U_1\otimes \ldots \otimes U_n \ | \ U_i \in \cal C_1\} < \cal C_n \ .\]
\end{Def}

\noindent By definition, the local Clifford group $\cal C_n^{\text{loc}}$ acts on the Pauli group $\cal P_n$ by conjugacy

\begin{equation}\label{conjugacy}
\begin{matrix}
\cal C_n^{\text{loc}}\times \cal P_n & \longrightarrow & \cal P_n \\
(U_1\otimes \ldots \otimes U_n, A_1\otimes \ldots \otimes A_n) & \mapsto & (U_1A_1U_1^\dagger)\otimes \ldots \otimes (U_nA_nU_n^\dagger)
\end{matrix} \ .
\end{equation}

It is known that $\cal C_1^{\text{loc}}=\cal C_1=\langle H, \sqrt{Z}\rangle$ and $\cal C_n^{\text{loc}}=\langle H_j, \sqrt{Z}_k\rangle$. Moreover, $\cal P_1 \subset \langle H , \sqrt{Z}\rangle=\cal C_1$: indeed $Z=(\sqrt{Z})^2 \in \cal C_1$, $X=HZH^\dagger\in \cal C_1$ and $Y= -iXZ \in \cal C_1$. Let us explicit the action of $\cal C_1$ on the elementary Pauli matrices: 
{\small \begin{equation}\label{C_1 on P_n}
	\begin{tikzpicture}[scale=2.5]
	\node(X) at (-1.5,0.5){$X$};
	\node(Z) at (0,0.5){$Z$};
	\node(Y) at (1.5,0.5){$Y$};
	\node(HX) at (-1.8,-0){$Z$};
	\node(SX) at (-1.2,-0){$Y$};
	\node(HZ) at (-0.3,-0){$X$};
	\node(SZ) at (0.3,-0){$Z$};
	\node(HY) at (1.2,0){$-Y$};
	\node(SY) at (1.8,0){$-X$};
	
	\path[font=\scriptsize,->, >= angle 90]
	(X) edge node [left] {$H\curvearrowright$} (HX)
	(X) edge node [right] {$\sqrt{Z}\curvearrowright$} (SX)
	(Z) edge node [left] {$H\curvearrowright$} (HZ)
	(Z) edge node [right] {$\sqrt{Z}\curvearrowright$} (SZ)
	(Y) edge node [left] {$H\curvearrowright$} (HY)
	(Y) edge node [right] {$\sqrt{Z}\curvearrowright$} (SY);
	\end{tikzpicture}
\end{equation}}

\subsection{The symplectic polar space $\cal W(2n-1,2)$ and $\cal I^n$}

\indent \indent  By Remark \ref{coordinates} we know that the Pauli group quotient $V_n$ is in one-to-one correspondence with the vector space $\mathbb F_2^{2n}$. Let us denote the binary projective space $\mathbb P(\mathbb F_2^{2n})$ by $\mathbb P_2^{2n-1}$ and let us fix the coordinates \eqref{articlecoordinates}. Thus we have
\begin{equation}\label{projlessoncoordinates}
\begin{matrix}
& \cal P_n & \dashrightarrow & \mathbb P_2^{2n-1} & \\
M= &\alpha Z^{\mu_1}X^{\nu_1}\otimes \ldots \otimes Z^{\mu_n}X^{\nu_n} &  \mapsto & [\mu_1: \ldots : \mu_n : \nu_1 : \ldots : \nu_n] & =P_M
\end{matrix}
\end{equation}
where the arrow is dashed since the above map is actually not defined in $\alpha I^{\otimes n}$.\\
\noindent Up to the coefficients $\alpha$'s (which are considered as global phases), this association completely describes $\cal P_n$ by projective points as a set, but it losts the commutation information we have in $\cal P_n$. Next we want to recover such information \cite{havlicek2009,levay2017}.

\indent Let $M_1,M_2 \in \cal P_n$ and let $P_1,P_2 \in \mathbb P_2^{2n-1}$ be the corresponding points. By a simple count it comes out that
\[M_1M_2 = M_2M_1 \iff \sum_{i=1}^n\big(\mu_i^{(1)}\nu_i^{(2)}-\mu_i^{(2)}\nu_i^{(1)}\big)=0 \ .\]
Consider the symplectic bilinear form $\langle \cdot , \cdot \rangle_J$ on $\mathbb F_2^{2n}$ (hence on $\mathbb P_2^{2n-1}$) where

\[ J=\begin{bmatrix} & I_n \\ -I_n & \end{bmatrix} \stackrel{\mathbb F_2}{=} \begin{bmatrix} & I_n \\
I_n & \end{bmatrix} \ .\]

\noindent Since $\langle P_1,P_2\rangle_J= \!^tP_1JP_2=\sum_{i=1}^n\big(\mu_i^{(1)}\nu_i^{(2)}-\mu_i^{(2)}\nu_i^{(1)}\big)$, the previous relation is equivalent to 
\[ M_1M_2 = M_2M_1 \iff \langle P_1, P_2 \rangle _J =0\]
that is {\em commuting Pauli group elements correspond to isotropic points with respect to $\langle\cdot,\cdot \rangle_J$}.

\begin{Def}
	The {\bf symplectic polar space} of rank $n$ over $\mathbb F_2$ (with respect to $\langle\cdot,\cdot\rangle_J$) is the set of (fully) isotropic subspaces of $\mathbb P_2^{2n-1}$
	\[ \cal W(2n-1,2)=\big\{ W \subset \mathbb P_2^{2n-1} \ \big| \ \forall P,Q \in W, \ \langle P,Q\rangle_J=0\big\} \ .\]
	\end{Def}

\indent Let $M \in \cal P_n$ and let $P\in \mathbb P_2^{2n-1}$ be its associated point: then $P$ defines the hyperplane in $\mathbb P_2^{2n-1}$
\[ H_P=\big\{ Q \in \mathbb P_2^{2n-1} \ \big| \ \langle P,Q\rangle_J=0\big\} \ .\]
Clearly, this hyperplane is not fully isotropic. However, we can extend this construction to any set of Pauli group elements with the only condition that they mutually commute:
\[\begin{matrix}
\cal P_n & \dashrightarrow & \mathbb P_2^{2n-1} & \stackrel{\vee}{\rightarrow} & \big(\mathbb P_2^{2n-1}\big)^\vee \\
M_1, \ldots , M_k & \mapsto & P_1, \ldots , P_k & \mapsto & H_{P_1,\ldots,P_k}=\big\{ Q \in \mathbb P_2^{2n-1} \ \big| \ \langle P_i,Q\rangle_J=0, \ \forall i =1:k\big\}
\end{matrix}\]

\noindent The condition for the $M_i$'s to be mutually commuting implies that the $P_i$'s are two-by-two isotropic, but the subspace $H_{P_1,\ldots,P_k}$ is not fully isotropic in general.\\
\indent We know that, for any $P\in \mathbb P_2^{2n-1}$, $H_P$ is a hyperplane, hence it has (projective) dimension $2n-2$. Generally, given $P_1, \ldots , P_k \in \mathbb P_2^{2n-1}$ two-by-two isotropic, the subspace $H_{P_1,\ldots , P_k}$ does not have dimension $2n-1-k$, but it holds so if we start from $k$ mutually commuting Pauli group elements $M_1,\ldots, M_k$ with the additional condition to be {\em independent}, in the sense
	\[ M_1^{c_1}\cdot \ldots \cdot M_k^{c_k}=I^{\otimes n} \iff \forall i=1:k, \ c_i=0 \ .\]
	
\begin{prop}\label{subspacesascolumns}
	Let $M_1, \ldots, M_k \in \cal P_n$ be $k$ Pauli group elements and let $P_1,\ldots , P_k \in \mathbb P_2^{2n-1}$ be their associated points. Then the following are equivalent:
	\begin{enumerate} 
	    \item[$(i)$] $M_1,\dots,M_k$ are mutually commuting and independent Pauli group elements,
		\item[$(ii)$] the matrix $S=\big[P_1|\ldots |P_k\big]\in \Mat_{2n\times k}(\mathbb F_2)$ has  rank $k$ and it holds \space  $\!^tSJS=0$;
		\item[$(iii)$] $\dim_{\mathbb P}H_{P_1,\ldots , P_k}=2n-1-k$ and $\text{Col}(S)\subset H_{P_1,\ldots , P_k}= \text{Col}(S)^\perp$, where Col$(S)\subset \mathbb P_2^{2n-1}$ is the subspace generated by the columns of $S$ and Col$(S)^\perp$ is its orthogonal for the symplectic form $\langle , \rangle_J$.
		\end{enumerate}
	\end{prop}
	\begin{proof}
	 $(i) \Leftrightarrow (ii)$ follows from the definitions: the condition of being independent is equivalent to require that the matrix $S$ has rank $k$ while the condition of being mutually commuting translates to $^tSJS=0$.
	 $(ii) \Leftrightarrow (iii)$: The condition on the dimension is equivalent to the fact that $S$ has rank $k$. By definition, $H_{P_1,\dots,P_n}=\text{Col}(S)^\perp$ and $\text{Col}(S)\subset\text{Col}(S)^\perp$ is equivalent to the condition $^tSJS=0$.
	\end{proof}

\begin{obs}\label{symplectic and quadric}
	The condition $\text{Col}(S)\subset \text{Col}(S)^\perp$ imposes some restrictions on $k$ for the maximal number of mutually communting and independent elements $M_1,\dots,M_k$. Indeed, $\text{Col}(S)\subset \text{Col}(S)^\perp$ implies that $k$ should satisfy $k-1\leq 2n-1-k$, i.e. $k\leq n$.
	\end{obs}
	
By Remark \ref{symplectic and quadric}, it follows that in order to actually reach out subspaces of type $H_{P_1,\ldots , P_k}$ (with $P_i$'s mutually commuting and independent) which are fully isotropic (i.e. $H_{P_1,\ldots P_k}\in \cal W(2n-1,2)$) we need to impose $\dim H_{P_1,\ldots , P_k}=2n-1-k = n-1$, that is $k = n$. The $(n-1)$-dimensional fully isotropic subspaces of $\mathbb P_2^{2n-1}$ are known as generators of $\mathcal{W}(2n-1,2)$  and we denote their set by
\[ \cal I^n= \big\{ W \in \cal W(2n-1,2) \ \big| \ \dim_{\mathbb P}W=n-1\big\} \ .\]
By Proposition \eqref{subspacesascolumns}$(ii)$, it follows that every $W \in \cal I^n$ is of the form $H_{P_1,\ldots,P_n}$ where $P_1, \ldots, P_n$ come from mutually commuting and independent Pauli group elements $M_1,\ldots, M_n$ via the following correspondence:
\begin{equation}\label{Pauli-I^n}
\begin{matrix}
\underbrace{M_1, \ldots , M_n}_{\text{indep. \& commut.}} & \Leftrightarrow & \underbrace{P_1,\ldots , P_n}_{\text{lin. indep. in }\mathbb F_2^{2n}} & \Leftrightarrow & \underbrace{H_{P_1,\ldots , P_n}}_{\dim_{\mathbb P}=n-1}
\end{matrix} \ .
\end{equation}

\begin{Def}
	A {\bf maximal ($n$-fold) abelian subgroup} $\cal S_{M_1,\ldots, M_n}<\cal P_n$ is a subgroup of $\cal P_n$ generated by $n$ independent and mutually commuting elements $M_1,\ldots , M_n$.
	\end{Def}

\noindent Let us denote by $\scr S(\cal P_n)$ the set of maximal abelian subgroups in $\cal P_n$. Thus, in the previous notations, we can reinterpret the correspondence \eqref{Pauli-I^n} as 

\begin{equation}\label{stab-In} 
\begin{matrix}
\cal S_{M_1,\ldots , M_n} \in \scr S(\cal P_n) \ \longleftrightarrow \ H_{P_1,\ldots, P_n} \in \cal I_n
\end{matrix} \ .
\end{equation}

\subsection{The Lagrangian Grassmannian $\LG_{\mathbb F_2}(n,2n)$ and $\cal Z_n$}

\indent \indent We recall that the Grassmannian $\Gr_{\mathbb F_2}(n,2n)$ is the set of $n$-dimensional subspaces of $\mathbb F_2^{2n}$: we will write $\Gr(n,2n)$ by omitting the ground field $\mathbb F_2$. The Grasssmannian gains the structure of projective variety via the Pl\"{u}cker embedding \cite{harris2013}
\begin{equation}\label{pluckerembedding} \begin{matrix}
\Pl: & \Gr(n,2n) & \hookrightarrow & \mathbb P\bigg(\bigwedge^n \mathbb F_2^{2n}\bigg) &\simeq \mathbb P_2^{\binom{2n}{n}-1}\\
& \langle v_1, \ldots , v_n\rangle & \mapsto & [v_1\wedge \ldots \wedge v_n]
\end{matrix} \ .\end{equation}
By abuse of notation, we simply look at $\Gr(n,2n)$ as a subset of $\mathbb P \big(\bigwedge^n\mathbb F_2^{2n}\big)$. It is useful to describe the Grassmannian by its parameterizations in $\mathbb P_2^{\binom{2n}{n}-1}$ on the standard open subsets
\begin{equation}\label{grassmannian chart} U_{\{i_1,\ldots,i_n\}}:=\left\{\langle v_1,\ldots, v_n\rangle \in \Gr(n,2n) \ \bigg| \ \det {\scriptsize \begin{bmatrix} v_{1,i_1} & \cdots & v_{n,i_1}\\
	\vdots & & \vdots \\
	v_{1,i_n} & \cdots & v_{n,i_n} \end{bmatrix}}\neq 0\right\} \end{equation}
for any index subset $\{i_1,\ldots, i_n\}\subset \{1,\ldots , 2n\}$. For simplicity, we recall what happens in the open subset $U_{\{1,\ldots, n\}}$: given $W=\langle v_1 , \ldots , v_n \rangle \in \Gr(n,2n)$, we can rewrite the vectors $v_i$'s with respect to the canonical basis $(e_1,\ldots, e_{2n})$ of $\mathbb F_2^{2n}$ as
\begin{equation}\label{pluckercoordinates} v_i= e_i + \sum_{j=1}^n a_{ji}e_{n+j}\end{equation}
and, by putting the basis vectors in columns, this gives the matrix
\[ \Bigg[ v_1 \ \bigg| \ \ldots \ \bigg| \ v_n \Bigg]  = {\tiny \begin{bmatrix}
1 & & \\
& \ddots & \\
& & 1\\
a_{11} & \ldots & a_{1n} \\
\vdots & \ddots & \vdots \\
a_{n1} & \ldots & a_{nn}
\end{bmatrix}}= \begin{bmatrix}

I_n\\

A\\

\end{bmatrix} \in \Mat_{2n\times n}(\mathbb F_2) \ .\]
\begin{Def}
	The above matrix $A\in \Mat_{n}(\mathbb F_2)$ is called the {\bf Pl\"{u}cker} (coordinates) {\bf matrix} with respect to the open subset $U_{\{1,\ldots, n\}}$ of the subspace $\langle v_1 , \ldots , v_n\rangle$, and we refer to its (ordered) columns as the {\bf Pl\"{u}cker basis} in $U_{\{1,\ldots, n\}}$.
	\end{Def}

\begin{obs}
	Given a general index subset $\{t_1,\ldots, t_n\}\subset \{1,\ldots, 2n\}$, for any subspace $W \in U_{\{t_1,\ldots, t_n\}}$ we can always consider its Pl\"{u}cker basis in $U_{\{t_1,\ldots, t_n\}}$
	\begin{equation}\label{plucker coordinates in general chart} 
		v_i= e_{t_i}+ \sum_{j=1}^n a_{ji}e_{s_j}
		\end{equation}
	where $\{s_1,\ldots, s_n\}=\{1,\ldots ,2n\}\setminus \{t_1,\ldots, t_n\}$: in particular, $W$ is described by a $2n\times n$ matrix having the $i$-th identity row in the $t_i$-h row. However, notice that \eqref{pluckercoordinates} and \eqref{plucker coordinates in general chart} are equivalent up to permutation of the standard basis vectors $e_i$'s.\\
	It is also worth remarking that a given subspace $W \in \Gr(n,2n)$ may lie in different open subsets of the form $U_I$, and thus we can consider its Pl\"{u}cker basis (and Pl\"{u}cker matrix) with respect to any subset $U_I$ containing it.
	\end{obs}

\begin{obs}
	For simplicity, in the following we will work in the open subset $U_{\{1,\ldots, n\}}$, but every definition and construction can be readapted to any standard open subset $U_{\{t_1,\ldots, t_n\}}$. 
	\end{obs}

By taking for any subspace $W \in U_{\{1,\ldots, n\}}$ its Pl\"{u}cker basis $(v_1,\ldots, v_n)$ (in $U_{\{1,\ldots,n\}}$) as \eqref{pluckercoordinates}, one gets the following projective coordinates 
\[\begin{matrix}
U_{\{1,\ldots, n\}} & \stackrel{\simeq}{\longrightarrow} & \mathbb P_2^{\binom{2n}{n}-1}\\
[v_1 \wedge \ldots \wedge v_n] & \mapsto & \big[1 : a_{ji} : A_{\{i,j\}\{s,t\}}: \ldots : A_{\{i_1,\ldots, i_k\}\{j_1,\ldots, j_k\}} : \ldots : \det A \big]
\end{matrix}\]
where $A_{\{i_1,\ldots, i_k\}\{j_1,\ldots, j_k\}}$ is the (determinant of the) $k\times k$ minor of $A$ given by the rows $i_1, \ldots , i_k$ and the columns $j_i,\ldots, j_k$. 

\begin{fact}
	Any open subset $U_I \subset \Gr(n,2n) \hookrightarrow \mathbb P\big(\bigwedge^n\mathbb F_2^{2n}\big)$ is parameterized by all minors of $n\times n$ matrices with coefficients in $\mathbb F_2$.
\end{fact}

\begin{obs}\label{othercoordGr}
	The parameterization of the open subset $U_{\{1,\ldots, n\}}\subset \Gr(n,2n)$ by all minors (of any size) of $n\times n$ matrices (as well as the one of any subset $U_I$) can be seen as induced by all {\em maximal} minors (of size $n\times n$) of $2n\times n$ matrices of the form ${\footnotesize \begin{bmatrix}I_n \\ A\end{bmatrix}}$ as $A\in \Mat_{n}(\mathbb F_2)$ varies: indeed, it is enough to project onto the coordinates given by minors fully contained in $A$.
\end{obs}

\indent Let us restrict the Pl\"{u}cker embedding \eqref{pluckerembedding} to $\cal I^n\subset \cal W(2n-1,2)$: more formally, we have to apply the Pl\"{u}cker embedding to the set of $n$-dimensional vector subspaces of $\mathbb F_2^{2n}$ whose projectivizations are in $\cal I_n$.

\begin{Def}
	The {\bf Lagrangian Grassmannian} $\LG_{\mathbb F_2}(n,2n)$ is the image via the Pl\"{u}cker embedding of $\cal I^n$
	\[ \LG_{\mathbb F_2}(n,2n)=\Pl(\cal I^n) \subset \mathbb P\bigg(\bigwedge^n\mathbb F_2^{2n}\bigg) \ .\]
\end{Def}

\noindent Thus the Lagrangian Grassmannian is the projective variety parameterizing all $(n-1)$-dimensional (fully) isotropic subspaces of $\mathbb P_2^{2n-1}$: again, we write $\LG(n,2n)$ by omitting the ground field.

\begin{Def}
	Since $\LG(n,2n)\subset \Gr(n,2n)$, for any standard open subset $U_I\subset \Gr(n,2n)$ as in \eqref{grassmannian chart} we denote its restriction to the Lagrangian Grassmannian by
	\begin{equation}\label{charts in LG}
		LU_I := U_I \cap \LG(n,2n) \ .
		\end{equation}
	\end{Def}

\indent Our next goal is to find parameterizations of $\LG(n,2n)$ in $\mathbb P_2^{\binom{2n}{n}-1}$ \cite{landsberg2012}.\\
\hfill\break
\noindent {\bf Setting:} Let us fix the coordinates \eqref{articlecoordinates} in $\mathbb P_2^{2n-1}$ and the symplectic form $\langle\cdot , \cdot \rangle_J$.\\
\hfill\break
Given $(v_1,\ldots , v_n)$ a Pl\"{u}cker basis of $H_{P_1,\ldots , P_n}\in \cal I^n\cap U_{\{1,\ldots, n\}}$, by isotropicity we have $\langle v_i,v_j\rangle_{J}=0$ for all $i,j=1:n$: then from \eqref{pluckercoordinates} we get
\begin{align*} \langle v_i,v_j \rangle_{ J}=0 & \iff \bigg\langle e_i+\sum_{k=1}^n a_{ki}e_{n+k}, e_j+\sum_{l=1}^n a_{lj}e_{n+l}\bigg\rangle_{J}=0 \\
& \iff  \bigg\langle e_i+ \sum_{k=1}^n a_{ki}e_{n+k}, \sum_{l=1}^n a_{lj}e_{l} + e_{n+j}\bigg\rangle_I=0\\
& \iff a_{ij} + a_{ji}=0 \stackrel{\mathbb F_2}{\iff} a_{ji}=a_{ij}
\end{align*}
that is the Pl\"{u}cker matrix $A$ of $H_{P_1,\ldots , P_n}\in \cal I^n$ in $U_{\{1,\ldots, n\}}$ is symmetric.

\begin{fact}
	Any standard open subset $LU_I\subset \LG(n,2n)$ is parameterized by all minors of $n\times n$ symmetric matrices with coefficients in $\mathbb F_2$. For instance, the coordinates in the standard open subset  $LU_{\{1,\ldots,n\}}$ are 
	\begin{equation}\label{symmmatrices} \begin{matrix}
	A=(a_{ji})_{i,j} \in \Sym^2(\mathbb F_2^{n}) & \mapsto & \big[1 : a_{ji} : A_{\{i,j\}\{s,t\}}: \ldots :  \det A \big] \in LU_{\{1,\ldots, n\}}
	\end{matrix}\end{equation}
\end{fact}

\indent Among all minors $A_{\{i_1\ldots i_k\}\{j_1\ldots j_k\}}$ of a $n\times n$ symmetric matrix $A\in \Sym^2(\mathbb F_2^{n})$, we can restrict ourselves to consider the principal ones, that are the ones such that $i_1=j_1 , \ldots , i_k=j_k$: we denote them by
\[ A_{[i_1,\ldots , i_k]}=\det {\footnotesize \begin{bmatrix}
	a_{i_1,i_1} & \ldots & a_{i_1,i_k}\\
	\vdots & \ddots & \vdots \\
	a_{i_k, i_1} & \ldots & a_{i_k,i_k}
	\end{bmatrix}} \ .\]
Clearly, the principal minors appear in the coordinates of the Lagrangian Grassmannian: in a very na\"{\i}f way, we can write $\{A_{[i_1,\ldots , i_k]}\}\subset \{A_{\{i_1,\ldots, i_k\}\{j_1,\ldots , j_k\}}\}$. Thus it makes sense to consider the {\em rational} projection of the Langrangian Grassmannian $\LG(n,2n)$ onto the coordinates represented by principal minors: since the number of principal minors of a $n\times n$ matrix is $\sum_{q=0}^n\binom{n}{q}=2^n$, we have that such rational projection has values in $\mathbb P_2^{2^n-1}$. For instance, in the open subset $LU_{\{1,\ldots, n\}}\subset \LG(n,2n)$ we have

\[ \begin{matrix}
LU_{\{1,\ldots, n\}}\subset \mathbb P_2^{\binom{2n}{n}-1} & \stackrel{\pi_{|_{LU_{\{1,\ldots, n\}}}}}{\dashrightarrow} & \mathbb P_2^{2^n-1}\\
\big[1 : a_{ji} : A_{\{i,j\}\{s,t\}}: \ldots :  \det A \big] & \mapsto & \big[   1: a_{ii} : A_{[i,j]} : \ldots : \det A\big]
\end{matrix} \ .\]

\begin{Def}\cite{oeding2009}
	The image of $\LG(n,2n)$ via the rational projection $\pi$ is the {\bf variety of principal minors} of $n\times n$ {\em symmetric} matrices. We denote it by
	\[ \cal Z_n = \pi\big(\LG(n,2n)\big) \subset \mathbb P_2^{2^n-1} \ .\]
	Moreover, for any standard open subset $LU_{I}\subset \LG(n,2n)$ we denote its image in $\cal Z_n$ by
	\[ \cal ZU_I:=\pi(LU_I) \ .\]
\end{Def}

\noindent Over a generic field $\mathbb K$, the projection $\pi$ is just surjective, but over $\mathbb F_2$ it is injective too.

\begin{prop}[\cite{holweck2014}]\label{LGisoZ}
	The projection $\pi : \LG_{\mathbb F_2}(n,2n)\twoheadrightarrow \cal Z_n \subset \mathbb P_2^{2^n-1}$ is a bijection.
\end{prop}
\begin{proof}
	Fix an open subset $LU_{I}\subset \LG(n,2n)$: then any subspace in it is parameterized by a certain $A \in \Sym^2(\mathbb F_2^n)$. The off-diagonal entries of $A$ are determined by the $2\times 2$ principal minors: indeed
	\[ A_{[i,j]}=a_{ii}a_{jj}-a_{ij}^2 \stackrel{\mathbb F_2}{\implies} a_{ij}=a_{ij}^2=A_{[i,j]}-a_{ii}a_{jj} \ .\]
\end{proof}

\indent Let us recap all the objects we have worked with so far in a unique diagram:

\begin{equation}\label{finalobjects}
\begin{tikzpicture}[scale=2.5]
\node(1) at (-1.5,1.4){$\mathbb P_2^{2n-1}$};
\node(2) at (0.2,1.4){$\mathbb P_2^{\binom{2n}{n}-1}$};
\node(3) at (2,1.4){$\mathbb P_2^{2^n-1}$};

\node(I) at (-1.5,1.1){$\cal I^n$};
\node(LG) at (0.2,1.1){$\LG_{\mathbb F_2}(n,2n)$};
\node(Z) at (2,1.1){$\cal Z_n$};

\node(I) at (-1.5,0.8){$U_{\{1,\ldots, n\}}\cap \cal I^n$};
\node(LG) at (0.2,0.8){$LU_{\{1,\ldots, n\}}$};
\node(Z) at (2,0.8){$\cal ZU_{\{1,\ldots, n\}}$};

\node(H) at (-1.5,0.5){$H_{P_1,\ldots , P_n}$};
\node(A) at (0.2,0.5){$ \big[1: a_{ij} : A_{\{i,j\}\{s,t\}} : \ldots \big] $};
\node(D) at (2,0.5){$ \big[ 1: a_{ii} : A_{[i,j]}: \ldots  \big]$};

\node(Hpl) at (-1.5,-0.2){$\bigg\langle {\tiny\begin{bmatrix} e_1 \\ -- \\ a_{11} \\ \vdots \\ a_{n1} \end{bmatrix}}, \ldots , {\tiny \begin{bmatrix} e_n \\ -- \\ a_{1n} \\ \vdots \\ a_{nn}\end{bmatrix}} \bigg\rangle$};
\node(wedge) at (0.2,-0.2){${\tiny\begin{bmatrix} e_1 \\ -- \\ a_{11} \\ \vdots \\ a_{n1} \end{bmatrix}}\wedge \ldots \wedge {\tiny \begin{bmatrix} e_n \\ -- \\ a_{1n} \\ \vdots \\ a_{nn}\end{bmatrix}}$};

\path[font=\scriptsize, >= angle 90]
(I) edge [->] node [above] {$\Pl$} (LG)
(LG) edge [->] node [above] {$\pi \simeq$} (Z)

(H) edge [-,dashed] node [left] {$=$} (Hpl)
(Hpl) edge [|->] node [above] {$\Pl$} (wedge)
(wedge) edge [->,dashed] node [] {} (A)
(A) edge [|->] node [above] {$\pi$} (D)
;
\end{tikzpicture}
\end{equation}

\section{Orbits in $\cal Z_n$ induced by the action $\cal C_n^{\text{loc}} \curvearrowright \cal P_n$}\label{orbitZn}

\indent \indent In this section we discuss the correspondence between the action of the group $\cal C_n^{\text{loc}}$ on $\mathcal{W}(2n-1,2)$ and the action of the group  $\text{SL}(2,\mathbb{F}_2)^{\times n}$ on $\cal Z_n$.

\subsection{$\Sp^{\text{loc}}_{2n}(\mathbb F_2)$ acting on $\cal W(2n-1,2)$ and $\cal I^n$}

\indent \indent We recall the conjugacy action \eqref{conjugacy} of the local Clifford group on the Pauli group
\[
\begin{matrix}
\cal C_n^{\text{loc}}\times \cal P_n & \longrightarrow & \cal P_n \\
(U_1\otimes \ldots \otimes U_n, A_1\otimes \ldots \otimes A_n) & \mapsto & (U_1A_1U_1^\dagger)\otimes \ldots \otimes (U_nA_nU_n^\dagger)
\end{matrix}
\]
and its induced action on the maximal abelian subgroups
\begin{equation}\label{C_n^loc on S(P_n)}
\begin{matrix}
\cal C_{n}^{\text{loc}} \ \times \ \mathscr S(\cal P_n) & \longrightarrow & \mathscr S(\cal P_n) \\
(\ U \ , \ \langle M_1,\ldots, M_n\rangle \ ) & \mapsto & \langle UM_1U^\dagger, \ldots , UM_nU^\dagger \rangle
\end{matrix} \ .
\end{equation}
\noindent We are interested in understanding how this action translates onto the symplectic polar space $\cal W(2n-1,2)$, in particular onto $\cal I^n$, via the correspondence \eqref{Pauli-I^n}.

\begin{obs}
	We have to underline that the translated action on $\cal W(2n-1,2)$ will depend on the choice of coordinates we take in $\mathbb P_2^{2n-1}$. From now on we will work with the choice of coordinates \eqref{articlecoordinates} together with the symplectic form $J$, but all the results also hold with respect to the choice of coordinates in \eqref{lessoncoordinates} and the symplectic form over $\mathbb F_2$
	\[J'={\tiny\begin{bmatrix}
		0 & 1 \\
		1 & 0 \\
		  &   & \ddots \\
		  &   &        & 0 & 1 \\
		  &   &        & 1 & 0
		\end{bmatrix}} \ .\]
	\end{obs}

\indent Fix the coordinates \eqref{articlecoordinates} in $\mathbb P_2^{2n-1}$ and the symplectic form $\langle\cdot , \cdot \rangle_J$. The first natural step is to translate the conjugacy action $\cal C_n^{\text{loc}} \curvearrowright \cal P_n$ into linear transformations of $\mathbb P_2^{2n-1}$.

\begin{obs}\label{C_n in Sp}
	By definition, any element of the Clifford group $\cal C_n$ induces an automorphism of $\cal P_n$, hence an automorphism of $V_n=\cal P_n/Z(\cal P_n)$: but automorphisms of $V_n\simeq \mathbb F_2^{2n}$ are linear maps and they preserve commutators, hence also the symplectic form $J$ on $\mathbb F_2^{2n}$ is preserved. It follows that there exists a well-defined homomorphism
	\[ \begin{matrix}
		\cal C_n & \longrightarrow & \Sp(\mathbb F_2^{2n}) & \subset \GL(\mathbb F_2^{2n})\\
		g & \mapsto & \hat g
		\end{matrix}\]
	such that, given $M \in V_n\simeq \mathbb F_2^{2n}$ and $\tilde M \in \cal P_n$ any lifting of $M$, the action of $\hat g$ on $M$ is 
	\[ \hat g \cdot M = \overline{g \tilde M g^{-1}} \in V_n \ .\]
	\noindent The homomorphism $\cal C_n\rightarrow \Sp(\mathbb F_2^{2n})$ is surjective, since the symplectic group is spanned by symplectic transvections \cite[Sec. II.B]{transvection}, but it is not injective: its kernel is exactly the Pauli group $\cal P_n$ \cite{koenig}, thus one has the isomorphism $\cal C_n/\cal P_n \simeq \Sp(\mathbb F_2^{2n})$.
	\end{obs}

\indent The homomorphism $\cal C_n \rightarrow \Sp(\mathbb F_2^{2n})$ in Remark \ref{C_n in Sp} restricts to an homomorphism
	\begin{equation}\label{C_n^loc in Sp}
		\begin{matrix}
			\cal C_n^{\text{loc}} & \longrightarrow & \Sp(\mathbb F_2^{2n})\\
			U & \mapsto & \tilde U
			\end{matrix}
		\end{equation}
	The elements in the image of the above restriction are of the form
	\begin{equation}\label{locsymplform} \tilde U={\tiny \begin{bmatrix}
		a_1 & & & b_1 & & \\
		& \ddots & & & \ddots & \\
		& & a_n & & & b_n\\
		c_1 & & & d_1 & & \\
		& \ddots & & & \ddots & \\
		& & c_n & & & d_n\\ 
		\end{bmatrix}} \in \Sp(\mathbb F_2^{2n}) \ , \end{equation}
	as one can check by looking at the action of $\cal C_n^{\text{loc}}$ on the Pauli elements: by applying a given $U=U_1\otimes \ldots \otimes U_n \in \cal C_n^{\text{loc}}$ to the generators of $\cal P_n$ in \eqref{generators of P_n}, we get
	{\footnotesize \begin{align*}
		U\left( I\otimes \ldots \otimes \underbrace{Z}_{l-th} \otimes \ldots \otimes I \right) U^\dagger & = U_1U_1^\dagger \otimes \ldots \otimes U_lZU_l^\dagger \otimes \ldots \otimes U_nU_n^\dagger = I \otimes \ldots \otimes \underbrace{Z^{a_l}X^{c_l}}_{l-th}\otimes \ldots \otimes I \ , \\
		U\left( I\otimes \ldots \otimes \underbrace{X}_{s-th} \otimes \ldots \otimes I \right) U^\dagger & = U_1U_1^\dagger \otimes \ldots \otimes U_sXU_s^\dagger \otimes \ldots \otimes U_nU_n^\dagger = I \otimes \ldots \otimes \underbrace{Z^{b_s}X^{d_s}}_{s-th}\otimes \ldots \otimes I \ ,
		\end{align*}}
	\noindent hence
	{\small \begin{align}\label{from U to tilde U}
		U\left( I\otimes \ldots \otimes \underbrace{Z^{\mu_k}X^{\nu_k}}_{k-th} \otimes \ldots \otimes I \right) U^\dagger & = U_1U_1^\dagger \otimes \ldots \otimes U_k\left(Z^{\mu_k}X^{\nu_k}\right)U_k^\dagger \otimes \ldots \otimes U_nU_n^\dagger \\
		& = I \otimes \ldots \otimes \underbrace{Z^{a_k\mu_k+b_k\nu_k}X^{c_k\mu_k+d_k\nu_k}}_{k-th}\otimes \ldots \otimes I \ ,
		\end{align}}
	which in coordinates corresponds to 
	\[ \tilde U (0,\ldots, \mu_k ,\ldots, \nu_k, \ldots , 0)= (0,\ldots, \ a_k\mu_k+b_k\nu_k \ ,\ldots, \ c_k\mu_k+d_k\nu_k \ , \ldots , 0) \ . \]
	
\begin{obs}\label{tilde U_i in SL2}
	Equation \eqref{from U to tilde U} makes it clear that, given $U=U_1\otimes \ldots \otimes U_n \in \cal C_n^{\text{loc}}$, the submatrices $\tilde U_i:={\scriptsize \begin{bmatrix} a_i & b_i \\ c_i & d_i\end{bmatrix}}\in \Sp(2,\mathbb F_2)$ defining $\tilde U$ in \eqref{locsymplform} depend on the matrices $U_i \in \U(2,\mathbb C)$, but they are not the same (the former have coefficients in $\mathbb F_2$, the latter in $\mathbb C$). Moreover, since $\tilde U$ is symplectic, it holds $\det \tilde U_i=a_id_i-b_ic_i \neq 0$ for any $i=1:n$.
	\end{obs}

\begin{Def}
	We define the {\bf local symplectic group} $\Sp_{2n}^{\text{loc}}(\mathbb F_2)$ to be the image of the homomorphism \eqref{C_n^loc in Sp}, i.e.
	\[\Sp_{2n}^{\text{loc}}(\mathbb F_2) :=\big\{ S \in \Sp(\mathbb F_2^{2n}) \ \big| \ S \ \text{of the form } \eqref{locsymplform}  \big\} \ .\] 
	\end{Def}

\begin{obs}\label{C_n^loc not iso to Sp^loc}
	The surjective homomorphism $\cal C_n^{\text{loc}}\rightarrow \Sp_{2n}^{\text{loc}}(\mathbb F_2)$ is not injective since its kernel is the Pauli group $\cal P_n$. Indeed, for any two Pauli elements $U,A \in \cal P_n$ it holds $U A U^\dagger =\beta A$ for a suitable phase $\beta \in \{\pm 1, \pm i\}$, thus in the quotient space $V_n\simeq \mathbb F_2^{2n}$ one gets $\overline{UAU^\dagger}=\overline{\beta A}=\overline{A}$: it follows that $U \mapsto I_{2n}\in \Sp(\mathbb F_2^{2n})$. In particular, we get the isomorphism
	\[ \faktor{\cal C_n^{\text{loc}}}{\cal P_n} \simeq \Sp_{2n}^{\text{loc}}(\mathbb F_2) \ . \]
	\end{obs}

\indent From now on, we will denote by $\tilde U \in \Sp_{2n}^{\text{loc}}(\mathbb F_2)$ the symplectic matrix corresponding to $U \in \cal C_n^{loc}$. By \eqref{from U to tilde U} we can explicit the projective coordinates in $\mathbb P_2^{2n-1}$ under the action of $\tilde U \in \Sp_{2n}^{\text{loc}}(\mathbb F_2)$: 
	\[\tilde U \cdot [\mu_1:\ldots : \mu_n :\nu_1: \ldots : \nu_n] =\]
	\begin{equation}\label{Sp on points}
		[a_1\mu_1+b_1\nu_1: \ldots : a_n\mu_n+b_n\nu_n : c_1\mu_1+d_1\nu_1 : \ldots : c_n\mu_n+d_n\nu_n] \ .
 	\end{equation}

\noindent Next we translate the action onto $\cal W(2n-1,2)$: it immediately follows from the relation (even for non-isotropic subspaces)
\[\forall \tilde U \in \Sp_{2n}^{\text{loc}}(\mathbb F_2), \ \forall H_P\in (\mathbb P_2^{2n-1})^\vee, \ \ \  \tilde U\cdot H_P = H_{\tilde UP}\in (\mathbb P_2^{2n-1})^\vee \ .\]
Moreover, the local symplectic transformations preserve the dimensions of the (fully) isotropic subspaces in $\mathbb P_2^{2n-1}$, hence the subspace $\cal I^n$ is $\Sp_{2n}^{\text{loc}}(\mathbb F_2)$-invariant. Thus the action \eqref{conjugacy} $\cal C_n^{\text{loc}}\curvearrowright \cal P_n$ induces the action
\begin{equation}\label{Sp on I^n}
\begin{matrix}
\Sp_{2n}^{\text{loc}}(\mathbb F_2) \ \times \ \cal I^n & \longrightarrow & \cal I^n \\
(\ \tilde U \ , \ H_{P_1,\ldots , P_n} \ ) & \mapsto & H_{\tilde UP_1 , \ldots , \tilde UP_n}
\end{matrix} \ .
\end{equation}

\indent We can update the correspondence \eqref{Pauli-I^n} to the following one:
\begin{equation}\label{Cliff-Sp}
	\begin{tikzpicture}[scale=2.5]
	
	\node(M) at (-1.5,0.5){$\underbrace{M_1, \ldots , M_n}_{\text{indep. \& commut.}}$};
	\node(P) at (0,0.5){$\underbrace{P_1,\ldots , P_n}_{\text{indep. in }\mathbb F_2^{2n}}$};
	\node(H) at (1.5,0.5){$\underbrace{H_{P_1,\ldots , P_n}}_{\dim_{\mathbb P}=n-1}$};
	
	\node(UM) at (-1.5,-0.2){$UM_1U^\dagger, \ldots , UM_nU^\dagger$};
	\node(UP) at (0,-0.2){$\tilde UP_1, \ldots , \tilde UP_n$};
	\node(UH) at (1.5,-0.2){$H_{\tilde UP_1, \ldots , \tilde UP_n}$};	
	
	\path[font=\scriptsize,->, >= angle 90]
	
	(M) edge [dashed] node [above] {} (P)
	(P) edge [|->] node [above] {} (H)
	
	(UM) edge [dashed] node [below] {} (UP)
	(UP) edge [|->] node [above] {} (UH)
	
	(M) edge node [left] {$U\in \cal C_n^{\text{loc}}\curvearrowright$} (UM)
	(P) edge node [left] {$\tilde U\in \Sp_{2n}^{\text{loc}}(\mathbb F_2)\curvearrowright$} (UP)
	(H) edge node [right] {$\tilde U\in \Sp_{2n}^{\text{loc}}(\mathbb F_2)\curvearrowright$} (UH);
	\end{tikzpicture}
\end{equation}

\begin{obs}\label{SpGL}
	By definition and by Remark \ref{tilde U_i in SL2}, the local symplectic group $\Sp_{2n}^{\text{loc}}(\mathbb F_2)$ is isomorphic to $\SL(2,\mathbb F_2)^{\times n} \simeq \frak S_3^{\times n}$: 
	\[\begin{matrix}
	\cal C_n^{\text{loc}} & \stackrel{\simeq}{\longrightarrow} & \Sp_{2n}^{\text{loc}}(\mathbb F_2) & \stackrel{\simeq}{\longrightarrow} & \SL(2,\mathbb F_2)^{\times n} & \simeq \frak S_3^{\times n}\\
	U_1 \otimes \ldots \otimes U_n & \mapsto & \tilde U \ \text{as in \eqref{locsymplform}} & \mapsto & (\tilde U_1, \ldots , \tilde U_n)
	\end{matrix} \ .\]
\end{obs}

\begin{obs}\label{Jhat}
	The same arguments and results of this section hold if we fix the coordinates \eqref{lessoncoordinates} in $\mathbb F_2^{2n}$ and the symplectic form $\langle \cdot , \cdot \rangle_{J'}$. In this case a local Clifford element $U=U_1\otimes \ldots \otimes U_n\in \cal C_n^{\text{loc}}$ corresponds to a local symplectic transformation of the form
	\[ \tilde U'= {\tiny \begin{bmatrix}
			a_1 & b_1 \\
			& & a_2 & b_2 \\
			& & & & \ddots & \ddots \\
			& & & & & & a_n & b_n \\
			c_1 & d_1 \\
			& & c_2 & d_2 \\
			& & & & \ddots & \ddots \\
			& & & & & & c_n & d_n	
			\end{bmatrix}}  \in \Sp_{2n}^{\text{loc}}(\mathbb F_2) \ .\]
	However, we will keep working only in the coordinates setting  (\eqref{articlecoordinates}, $J$).
	\end{obs}

\subsection{$\Sp_{2n}^{\text{loc}}(\mathbb F_2)$ acting on $\LG_{\mathbb F_2}(n,2n)$}

\indent \indent Let us keep in mind the diagrams \eqref{Cliff-Sp} and \eqref{finalobjects}. We look for the action on $\LG(n,2n)$ induced by the action $\Sp_{2n}^{\text{loc}}(\mathbb F_2)\curvearrowright \cal I^n$ via the Pl\"{u}cker embedding: for simplicity, we describe the action on the standard open subset $LU_{\{1,\ldots,n\}}$ but all arguments apply to any standard open subset $LU_I$. \\
\hfill\break
\noindent {\bf Setting:} Fix the coordinates \eqref{articlecoordinates} in $\mathbb P_2^{2n-1}$ and the symplectic form $\langle\cdot , \cdot \rangle_{J}$.\\
\hfill\break
\indent By Remark \ref{othercoordGr}, we consider the parametrization of $LU_{\{1,\ldots, n\}}\subset \LG(n,2n)$ induced by all maximal minors of $2n\times n$ matrices of the form ${\scriptsize \begin{bmatrix}I_n\\ S\end{bmatrix}}$ as $S$ varies in $\Sym^2(\mathbb F_2^{n})$.\\
Consider a subspace $H_{P_1,\ldots,P_n}\in \cal I^n\cap U_{\{1,\ldots,n\}}$ (with Pl\"{u}cker basis as in \eqref{pluckercoordinates}) and a transformation $\tilde U \in \Sp_{2n}^{\text{loc}}(\mathbb F_2)$ such that
\[
H_{P_1,\ldots , P_n}=\bigg\langle {\tiny\begin{bmatrix} e_1 \\ -- \\ s_{11} \\ \vdots \\ s_{n1} \end{bmatrix}}, \ldots , {\tiny \begin{bmatrix} e_n \\ -- \\ s_{1n} \\ \vdots \\ s_{nn}\end{bmatrix}} \bigg\rangle \ \ \ , \ \ \ 
\tilde U= {\tiny \begin{bmatrix}
	a_1 & & & b_1 & & \\
	& \ddots & & & \ddots & \\
	& & a_n & & & b_n\\
	c_1 & & & d_1 & & \\
	& \ddots & & & \ddots & \\
	& & c_n & & & d_n\\ 
	\end{bmatrix}}=\begin{bmatrix} A& B \\ C& D \end{bmatrix}\]
where $S=(s_{ji})\in \Sym^2(\mathbb F_2^{n})$ and $a_id_i-b_ic_i\neq0$ for all $i=1:n$. By applying $\tilde U$ to the Pl\"{u}cker basis we obtain
\begin{equation}\label{Sp-Pl} 
\tilde U {\tiny\begin{bmatrix} e_j \\ -- \\ s_{1j} \\ \vdots \\ s_{nj} \end{bmatrix}} = \tilde U \bigg(e_j + \sum_{k=1}^n s_{kj}e_{n+k}\bigg) =  \sum_{k=1}^n (a_k\delta_{kj}+b_ks_{kj})e_k + (c_k\delta_{kj}+d_ks_{kj})e_{n+k} \end{equation}
where $\delta_{kj}$ is the Kronecker symbol: this is equivalent to the matrix product
\begin{align}\label{Sp-Pl-matrix}
& \tilde U \begin{bmatrix}
I_n\\
S
\end{bmatrix}= \begin{bmatrix} A+BS \\
C+DS\end{bmatrix}={\tiny \begin{bmatrix}
a_1+ b_1s_{11} & b_1s_{12} & \ldots & b_1s_{1n}\\
b_2s_{12} & a_2+b_2s_{22} & \ldots & b_2s_{2n}\\
    \vdots      &               & \ddots & \vdots \\
b_ns_{1n} & \ldots & & a_n+b_ns_{nn}\\
--- & --- & --- & ---\\
c_1+ d_1s_{11} & d_1s_{12} & \ldots & d_1s_{1n}\\
d_2s_{12} & c_2+d_2s_{22} & \ldots & d_2s_{2n}\\
\vdots      &               & \ddots & \vdots \\
d_ns_{1n} & \ldots & & c_n+d_ns_{nn}
\end{bmatrix}} \ .\end{align}

\noindent Thus $\tilde U(H_{P_1,\ldots , P_n})= \bigg\langle \tilde U{\tiny\begin{bmatrix} e_1 \\ -- \\ s_{11} \\ \vdots \\ s_{n1} \end{bmatrix}}, \ldots , \tilde U{\tiny \begin{bmatrix} e_n \\ -- \\ s_{1n} \\ \vdots \\ s_{nn}\end{bmatrix}} \bigg\rangle$ and, by applying the Pl\"{u}cker embedding,

\begin{equation}\label{Sp-Pl-wedge} 
\Pl\bigg( \tilde U(H_{P_1,\ldots , P_n})\bigg)= \Bigg[  \tilde U{\tiny\begin{bmatrix} e_1 \\ -- \\ s_{11} \\ \vdots \\ s_{n1} \end{bmatrix}} \wedge \ldots \wedge  \tilde U{\tiny \begin{bmatrix} e_n \\ -- \\ s_{1n} \\ \vdots \\ s_{nn}\end{bmatrix}}\Bigg] \ .\end{equation}

\noindent We conclude that the action \eqref{Sp on I^n} $\Sp_{2n}^{\text{loc}}\curvearrowright \cal I^n$ induces the action

\begin{equation}\label{Sp on LG}
\begin{matrix}
\Sp_{2n}^{\text{loc}}(\mathbb F_2) \ \times \ \LG(n,2n) & \longrightarrow & \LG(n,2n) \\
\big( \ \tilde U \ , \ [v_1 \wedge \ldots \wedge v_n] \ \big) & \mapsto & \big[\tilde Uv_1 \wedge \ldots \wedge \tilde Uv_n\big]
\end{matrix} \ .
\end{equation}

\begin{obs}\label{minors of A'}
	We must pay attention to the fact that the coordinates of \eqref{Sp-Pl-wedge} are given by all maximal (i.e. $n\times n$) minors of the Pl\"{u}cker matrix of \eqref{Sp-Pl-matrix} with respect to a suitable open subset $LU_I$. \\
	In general, the action by $\tilde U \in \Sp^{loc}_{2n}(\mathbb F_2)$ does not preserve the standard open subsets: given $H_{P_1,\ldots, P_n}\in LU_{\{1,\ldots,n\}}$, the image $\tilde U\left(H_{P_1,\ldots, P_n}\right)$ may lie in a different standard open subsets $LU_I$, and thus one has to consider the Pl\"{u}cker matrix of the latter subspace with respect to $LU_I$. For instance, in the notations of \eqref{Sp-Pl-matrix}, the subspace $\tilde U\left(H_{P_1,\ldots, P_n}\right)$ lies in $LU_{\{1,\ldots,n\}}$ if and only if $\det(A+BS)\neq 0$.
	\end{obs}


\subsection{$\Sp_{2n}^{\text{loc}}(\mathbb F_2)$ acting on $\cal Z_n$ as $\SL(2,\mathbb F_2)^{\times n}$}

\indent \indent The action \eqref{Sp on LG} translates into an action of $\Sp_{2n}^{\text{loc}}(\mathbb F_2)$ on $\cal Z_n$ 

\begin{equation}\label{Sp on Zn}
\Sp_{2n}^{\text{loc}}(\mathbb F_2) \ \times \ \cal Z_n  \longrightarrow  \cal Z_n
\end{equation}

\noindent via the projection $\pi : \LG(n,2n) \longrightarrow \cal Z_n$. By Remark \ref{SpGL} this action is equivalent, up to isomorphism, to an already known and natural action of $\SL(2,\mathbb F_2)^{\times n}$ on $\cal Z_n$:
the space $\mathbb P_2^{2^n-1}\simeq \mathbb P \big( \mathbb F_2^{2}\otimes \ldots \otimes \mathbb F_2^2 \big)$ is homogeneous under the natural action of $\SL(2,\mathbb F_2)^{\times n}$ and the following result shows that this action restricts to an action of $\SL(2,\mathbb F_2)^{\times n}$ on $\cal Z_n \subset \mathbb P_2^{2^n-1}$.

\begin{prop}[\cite{oeding2009}, Theorem III.14]\label{Zn invariant}
	$\cal Z_n$ is invariant under the natural action of $\SL(2,\mathbb F_2)^{\times n}$. Moreover, the action is represented by 
	\begin{equation}\label{SL representation}
	\begin{matrix} 
	\SL(2,\mathbb F_2)^{\times n} & \stackrel{\rho}{\longrightarrow} & \Sp\big(\mathbb F_2^{2n}\big) & \subset \GL\big(\mathbb F_2^{2n} \big)\\
	\bigg( {\tiny \begin{bmatrix} a_1 & b_1 \\ c_1 & d_1 \end{bmatrix}}, \ldots , {\tiny \begin{bmatrix} a_n & b_n \\ c_n & d_n \end{bmatrix}} \bigg) & \mapsto & {\tiny \begin{bmatrix}
		a_1 & & & b_1 & & \\
		& \ddots & & & \ddots & \\
		& & a_n & & & b_n\\
		c_1 & & & d_1 & & \\
		& \ddots & & & \ddots & \\
		& & c_n & & & d_n\\ 
		\end{bmatrix}} 
	\end{matrix}\end{equation}
	giving exactly the isomoprhism $\SL(2,\mathbb F_2)^{\times n} \simeq \Sp_{2n}^{\text{loc}}(\mathbb F_2)$.
\end{prop}

\begin{obs}
	Actually, in his PhD thesis \cite{oeding2009} L.Oeding proved the above result over $\mathbb C$, but it is straightforward that then it holds over $\mathbb F_2$ too.
	\end{obs}

\paragraph{Resume.}  We conclude this section by resuming how orbits in the Pauli group $\cal P_n$ under the action of the local Clifford group $\cal C_n^{\text{loc}}$ induce orbits in the variety of symmetric principal minors $\cal Z_n$ under the action of $\SL(2,\mathbb F_2)^{\times n}$.\\ 
\indent The local Clifford group $\cal C_n^{\text{loc}}$ acts on the Pauli group $\cal P_n$ by \eqref{conjugacy}
\[
\bigg(U_1\otimes \ldots \otimes U_n \bigg)\cdot  \bigg( \underbrace{Z^{\mu_1}X^{\nu_1}}_{A_1}\otimes \ldots \otimes \underbrace{Z^{\mu_n}X^{\nu_n}}_{A_n}\bigg) \ = \ 
U_1A_1U_1^\dagger \otimes \ldots \otimes U_nA_nU_n^\dagger
 \ .\]
\noindent By fixing the setting {\em \textquotedblleft coordinates - symplectic form\textquotedblright} (\eqref{articlecoordinates},J), the above action induces the action \eqref{Sp on I^n} of the local symplectic group $\Sp_{2n}^{\text{loc}}(\mathbb F_2)$ on the set $\cal I^n$ of $(n-1)$-dimensional (fully) isotropic subspaces of $\mathbb P_2^{2n-1}$ defined as follows: given $M_1,\ldots , M_n\in \cal P_n$ mutually commuting and independent such that
$M_i=Z^{\mu_1^{(i)}}X^{\nu_1^{(i)}}\otimes \ldots \otimes Z^{\mu_n^{(i)}}X^{\nu_n^{(i)}}$, and given their corresponding points
$P_i=[\mu_1^{(i)}: \ldots : \mu_n^{(i)}: \nu_1^{(i)}: \ldots : \nu_n^{(i)}] \in \mathbb P_2^{2n-1}$, it holds
\[
U\cdot H_{P_1,\ldots , P_n} \ = \ H_{UP_1, \ldots , UP_n} \]
for any $U\in \Sp_{2n}^{\text{loc}}(\mathbb F_2)$ as in \eqref{locsymplform}.
The action of $\Sp_{2n}^{\text{loc}}(\mathbb F_2)$ on $\cal I_n$ induces, via the Pl\"{u}cker embedding, the action  \eqref{Sp on LG} of the local symplectic group $\Sp_{2n}^{\text{loc}}(\mathbb F_2)$ on the Lagrangian Grassmannian $\LG(n,2n)$
\[U \cdot [v_1 \wedge \ldots \wedge v_n]= \big[Uv_1 \wedge \ldots \wedge Uv_n\big] \ . \]
Finally, the latter action translates into the action \eqref{Sp on Zn} of $\Sp_{2n}^{\text{loc}}(\mathbb F_2)$ on the variety of symmetric principal minors $\cal Z_n$, which is equivalent to the natural action of $\SL(2,\mathbb F_2)^{\times n}$ on $\cal Z_n$ via the representation \eqref{SL representation}.


\section{Correspondence between $\cal C_n^{\text{loc}}\rtimes \frak S_n \curvearrowright \scr S(\cal P_n)$ and $\SL(2,\mathbb F_2)^{\times n}\rtimes \frak S_n \curvearrowright \cal Z_n$}\label{orbitS}

\indent \indent In this section we extend the previous group actions to the semidirect product with the symmetric group $\mathfrak S_n$ in order to prove the first part of Theorem \ref{theo:main}.

\subsection{The actions $\cal C_n^{\text{loc}}\rtimes_\phi \mathfrak{S}_n \curvearrowright \cal P_n$ and $\Sp_{2n}^{\text{loc}}(\mathbb F_2)\rtimes_\varphi \mathfrak S_n \curvearrowright \mathbb F_2^{2n}$}

\indent \indent By definition of the $n$-fold Pauli group
{\small \begin{align}\label{elements in P_n}
	\cal P_n & =\left\{ A_1 \otimes \ldots \otimes A_n \ | \ A_i \in \cal P_1\right\}\\
	& =\big\{\pm Z^{\mu_1}X^{\nu_1}\otimes \ldots \otimes Z^{\mu_n}X^{\nu_n}, \pm i Z^{\mu_1}X^{\nu_1}\otimes \ldots \otimes Z^{\mu_n}X^{\nu_n} \ \big| \ \mu_i,\nu_i \in \{0,1\}\big\} \ ,
\end{align}}
\noindent there is a natural action of the symmetric group $\mathfrak{S}_n$ on $\cal P_n$ permuting the tensor entries:
\[ \forall \sigma \in \mathfrak S_n, \ \ \ \sigma \cdot (A_1 \otimes \ldots \otimes A_n) = A_{\sigma(1)}\otimes \ldots \otimes A_{\sigma(n)} \ .\]
\noindent Each permutation $\sigma \in \mathfrak S_n$ induces a transformation $\tilde \sigma\in U(2^n,\mathbb C)$ permuting the basis vectors, so that for any $\sigma \in \mathfrak S_n$ and for any $A_1\otimes \ldots \otimes A_n\in \cal P_n$ one gets
\begin{equation}\label{S_n acting on P_n} 
	\sigma \cdot (A_1\otimes \ldots \otimes A_n) = A_{\sigma(1)}\otimes \ldots \otimes A_{\sigma(n)}=\tilde \sigma (A_1\otimes \ldots \otimes A_n)\tilde \sigma^\dagger \ .
\end{equation}
Notice that $\tilde \sigma^\dagger =\tilde\sigma^{-1}$. In particular, the above action preserves $\cal P_n$, thus for any $\sigma \in \frak S_n$ it holds $\tilde \sigma \in \cal C_n=N_{U(2^n,\mathbb C)}(\cal P_n)$: it follows that there is an injective homomorphism 
\[ \begin{matrix}
\frak S_n & \hookrightarrow & \cal C_n & \subset U(2^n,\mathbb C)\\
\sigma & \mapsto & \tilde \sigma
\end{matrix}\]
which allows to identify $\frak S_n$ as a subgroup of the Clifford group $\cal C_n$. Moreover, the symmetric group $\frak S_n$ naturally acts on the local Clifford group $\cal C_n^{\text{loc}}$ by conjugacy 
\begin{equation}\label{S_n acting on C_n^loc} 
	\begin{matrix}
		\phi: & \frak S_n & \longrightarrow & \Aut\left(\cal C_n^{loc}\right)\\
		& \sigma & \mapsto & \big( \  \phi_\sigma: U \mapsto \ ^\sigma\!U=\tilde{\sigma} U \tilde{\sigma}^{-1} \ \big)
	\end{matrix} \ ,
\end{equation}
where $^\sigma\!U = \tilde\sigma U \tilde\sigma^{-1}= U_{\sigma(1)}\otimes \ldots \otimes U_{\sigma(n)}$ for any $U=U_1\otimes \ldots \otimes U_n \in \cal C_n^{\text{loc}}$.

\begin{obs}
	The action \eqref{S_n acting on C_n^loc} is actually well defined. Indeed, given $\sigma \in \frak S_n$, $U=U_1\otimes \ldots \otimes U_n \in \cal C_n^{\text{loc}}$ and $A_1\otimes \ldots \otimes A_n\in \cal P_n$, we have
	\begin{align*}
		^\sigma\!U(A_1 \otimes \ldots \otimes A_n)^\sigma\!U^\dagger & \ = (\tilde \sigma U \tilde \sigma^{-1})(A_1 \otimes \ldots \otimes A_n)(\tilde\sigma U \tilde\sigma^{-1})^\dagger\\
		& \ = (\tilde \sigma U \tilde \sigma^{-1})(A_1 \otimes \ldots \otimes A_n)(\tilde\sigma U^\dagger \tilde\sigma^{-1})\\
		& \stackrel{\eqref{S_n acting on P_n}}{=} \tilde \sigma U \left(A_{\sigma^{-1}(1)}\otimes \ldots \otimes A_{\sigma^{-1}(n)}\right)U^\dagger \tilde \sigma^{-1}\\
		& \ = \tilde \sigma \left(U_1A_{\sigma^{-1}(1)}U_1^\dagger\otimes \ldots \otimes U_nA_{\sigma^{-1}(n)}U_n^\dagger \right)\tilde \sigma^{-1}\\
		& \stackrel{\eqref{S_n acting on C_n^loc}}{=} U_{\sigma(1)}A_1U_{\sigma(1)}^\dagger \otimes \ldots \otimes U_{\sigma(n)}A_n U_{\sigma(n)}^\dagger \ \in \cal P_n \ .
	\end{align*}
\end{obs}

It follows that the subgroups $\cal C_n^{\text{loc}}$ and $\frak S_n$ (the second up to isomorphism) generate a subgroup in $\cal C_n$ which is isomorphic to the semidirect product $\cal C_n^{\text{loc}}\rtimes_\phi \frak S_n$
\[ \begin{matrix}
\cal C_{n}^{\text{loc}}\rtimes_\phi \frak S_n & \stackrel{\simeq}{\longrightarrow} & \left\langle \cal C_{n}^{\text{loc}} \ , \ \frak S_n\right\rangle & \subset \cal C_n\\
(U,\sigma ) & \mapsto & U \tilde\sigma
\end{matrix} \ ,\]
acting on $\cal P_n$ as follows
\begin{equation}\label{semidir acting on P_n}
	\begin{matrix}
		\big( \cal C_n^{\text{loc}}\rtimes_\phi \frak S_n\big) \times \cal P_n & \longrightarrow & \cal P_n\\
		\big( (U,\sigma), A_1\otimes \ldots \otimes A_n \big) & \mapsto & U \cdot \left( \sigma \cdot (A_1\otimes \ldots \otimes A_n) \right)
	\end{matrix}
\end{equation}
where, if $U=U_1 \otimes \ldots \otimes U_n$,
\begin{align*}
	U \cdot \left( \sigma \cdot (A_1\otimes \ldots \otimes A_n) \right) & = U\left(A_{\sigma(1)}\otimes \ldots \otimes A_{\sigma(n)}\right)U^\dagger \\
	& = U_{1}A_{\sigma(1)}U_{1}^\dagger \otimes \ldots \otimes U_{n}A_{\sigma(n)}U_{n}^\dagger \ .
	\end{align*}

\begin{obs}
	We recall that the operation in the semidirect product is
	\begin{equation}\label{operation in semidirect}
		(U, \sigma) \cdot_\phi (U',\tau)= \big(U \cdot \phi_\sigma(U'), \sigma\cdot \tau\big)
	\end{equation}
	and the following commutation rule holds
	\[\big(I,\sigma\big)\cdot_\phi (\phi^{-1}_{\sigma}(U),id) \ \stackrel{(a)}{=} \ (U,\sigma) \ \stackrel{(b)}{=} \ (U,id)\cdot_\phi (I,\sigma) \ .\]
	Since $\phi_\sigma^{-1}=\phi_{\sigma^{-1}}$, by equality $(a)$ it follows (in agreement with \eqref{semidir acting on P_n})
	\begin{align*} 
		(U,\sigma)\cdot (A_1\otimes \ldots \otimes A_n) & = \bigg(\big(I,\sigma\big)\cdot_\phi (\phi^{-1}_{\sigma}(U),id)\bigg) \cdot (A_1\otimes \ldots \otimes A_n)\\
		& = (I,\sigma)\cdot \bigg((^{\sigma^{-1}}U)(A_1 \otimes \ldots \otimes A_n)(^{\sigma^{-1}}U)^\dagger\bigg)\\
		& = (I,\sigma)\cdot \bigg(U_{\sigma^{-1}(1)}A_1U_{\sigma^{-1}(1)}^\dagger \otimes \ldots \otimes U_{\sigma^{-1}(n)}A_nU_{\sigma^{-1}(n)}^\dagger\bigg)\\
		& = U_1A_{\sigma(1)}U_1^\dagger \otimes \ldots \otimes U_nA_{\sigma(n)}U_n^\dagger \ .
	\end{align*}
\end{obs}

Since the elements in $\cal P_n$ are of the form $\alpha Z^{\mu_1}X^{\nu_1}\otimes \ldots \otimes Z^{\mu_n}X^{\nu_n}$ for $\alpha \in\{\pm 1, \pm i\}$, the action \eqref{S_n acting on P_n} of the symmetric group $\frak S_n$ on $\cal P_n$ can be equivalently described by
\begin{equation}\label{S_n acting on exponents}
	\sigma \cdot (\alpha Z^{\mu_1}X^{\nu_1}\otimes \ldots \otimes Z^{\mu_n}X^{\nu_n})= \alpha Z^{\mu_{\sigma(1)}}X^{\nu_{\sigma(1)}}\otimes \ldots \otimes Z^{\mu_{\sigma(n)}}X^{\nu_{\sigma(n)}} \ .
\end{equation}

Finally, from \eqref{Sp on points} we know that the action of a given $U=U_1\otimes \ldots \otimes U_n \in \cal C_n^{\text{loc}}$ on $ Z^{\mu_1}X^{\nu_1}\otimes \ldots \otimes Z^{\mu_n}X^{\nu_n}$ corresponds to an action of a certain $\left({\tiny \begin{bmatrix} a_1& b_1\\ c_1& d_1 \end{bmatrix}}, \ldots, {\tiny \begin{bmatrix} a_n& b_n\\ c_n& d_n \end{bmatrix}}\right)\in \SL(2,\mathbb F_2)^{\times n}$, where each ${\tiny \begin{bmatrix} a_i& b_i\\ c_i& d_i \end{bmatrix}}$ depends on $U_i$ and it acts on the coefficients $\mu_i,\nu_i$'s. In particular, in the same notation, we have
\begin{align*}
	U(Z^{\mu_1}X^{\nu_1}\otimes \ldots \otimes Z^{\mu_n}X^{\nu_n})U^\dagger & = U_1Z^{\mu_1}X^{\nu_1}U_1^\dagger\otimes \ldots \otimes U_nZ^{\mu_n}X^{\nu_n}U_n^\dagger\\
	& = Z^{a_1\mu_1+b_1\nu_1}X^{c_1\mu_1+d_1\nu_1}\otimes \ldots \otimes Z^{a_n\mu_n+b_n\nu_n}X^{c_n\mu_n+d_n\nu_n}
\end{align*}
and the action \eqref{semidir acting on P_n} of the semidirect product $\cal C_n^{\text{loc}}\rtimes_\phi\frak S_n$ on $\cal P_n$ can be rephrased as
\begin{equation}\label{semidir acting on exponents}
	(U,\sigma)\cdot (Z^{\mu_1}X^{\nu_1}\otimes \ldots \otimes Z^{\mu_n}X^{\nu_n}) \ = \ 
	U(Z^{\mu_{\sigma(1)}}X^{\nu_{\sigma(1)}}\otimes \ldots \otimes Z^{\mu_{\sigma(n)}}X^{\nu_{\sigma(n)}})U^\dagger  =
	\end{equation}	
\[
		Z^{a_1\mu_{\sigma(1)}+b_1\nu_{\sigma(1)}}X^{c_1\mu_{\sigma(1)}+d_1\nu_{\sigma(1)}}\otimes \ldots \otimes Z^{a_n\mu_{\sigma(n)}+b_n\nu_{\sigma(n)}}X^{c_n\mu_{\sigma(n)}+d_n\nu_{\sigma(n)}} \ .
		\]

\indent The above formula shows that the action of $\cal C_n^{\text{loc}}\rtimes_\phi \frak S_n$ on $\cal P_n$ induces an action on the vectors $(\mu_1, \ldots, \mu_n,\nu_1,\ldots , \nu_n) \in \mathbb F_2^{2n}$: more precisely, the action \eqref{semidir acting on exponents} induces an action on the quotient $V_n=\cal P_n / Z(\cal P_n)$ which is isomorphic to the vector space $\mathbb F_2^{2n}$ in the coordinates \eqref{articlecoordinates}.

\begin{obs}\label{C_n^loc semidir S_n injects in Sp(2n)}
	By Remark \ref{C_n in Sp} and \eqref{C_n^loc in Sp}, we know that there exists a group homomorphism $\cal C_n\rightarrow \Sp(\mathbb F_2^{2n})$ restricting to an homomorphism $\cal C_n^{\text{loc}}\twoheadrightarrow \Sp_{2n}^{\text{loc}(\mathbb F_2)}<\Sp(\mathbb F_2^{2n})$ with kernel $\cal P_n$. Actually, one can also consider the restriction to the subgroup $\cal C_n^{\text{loc}}\rtimes_\phi \frak S_n$ giving an homomorphism
	\[\cal C_n^{\text{loc}} \rtimes_\phi \frak S_n \longrightarrow  \Sp(\mathbb F_2^{2n})\]
	which is again not injective having kernel $\cal P_n \rtimes_\phi \mathfrak S_n$.
\end{obs}

\indent On the other hand, there is a well-defined action of $\frak S_n$ on $\mathbb F_2^{2n}$ given by
\[ \sigma \cdot (\mu_1, \ldots, \mu_n,\nu_1,\ldots , \nu_n) = (\mu_{\sigma(1)}, \ldots, \mu_{\sigma(n)},\nu_{\sigma(1)},\ldots , \nu_{\sigma(n)}) \ . \]
In particular, any $\sigma\in \frak S_n$ corresponds to a $S_\sigma\in \Sp(\mathbb F_2^{2n})$ such that 
\begin{align}\label{S_n on points}
		\sigma \cdot (\mu_1, \ldots, \mu_n,\nu_1,\ldots , \nu_n) & = (\mu_{\sigma(1)}, \ldots, \mu_{\sigma(n)},\nu_{\sigma(1)},\ldots , \nu_{\sigma(n)}) \\
		& = S_\sigma(\mu_1, \ldots, \mu_n,\nu_1,\ldots , \nu_n) \ , 
\end{align}
and this allows to identify $\frak S_n$ as a subgroup of $\Sp(\mathbb F_2^{2n})$.

\begin{obs}\label{form of S_sigma}
	The matrices $S_\sigma\in \Sp(\mathbb F_2^{2n})$ are of the form {\scriptsize $\begin{bmatrix} A_\sigma & 0 \\ 0 & A_\sigma \end{bmatrix}$} for some permutation matrix $A_\sigma \in \Sp(\mathbb F_2^n)$.
	\end{obs}

Finally, the symmetric group acts by conjugacy on $\Sp_{2n}^{\text{loc}}(\mathbb F_2)$ as follows
\begin{equation}\label{S_n acting on Sp^loc} 
	\begin{matrix}
		\varphi: & \frak S_n & \longrightarrow & \Aut\left(\Sp_{2n}^{\text{loc}}(\mathbb F_2)\right)\\
		& \sigma & \mapsto & \big( \  \varphi_\sigma: \tilde U \mapsto \ ^\sigma\!\tilde U=S_\sigma \tilde U S_{\sigma}^{-1} \ \big)
	\end{matrix} \ 
\end{equation}
which is well-defined since
\[
^t(^\sigma\!\tilde U)J(^\sigma\!\tilde U) = (^t\!S_\sigma^{-1})(^t\tilde U)\underbrace{(^t\!S_\sigma) J S_\sigma}_{=J} \tilde U S_\sigma^{-1}
= (^t\!S_\sigma^{-1})\underbrace{(^t\tilde U) J \tilde U}_{=J}  S_\sigma^{-1} \ = \ J \ .
\]
More precisely, if $\tilde U \in \Sp_{2n}^{\text{loc}}(\mathbb F_2)$ is as in \eqref{locsymplform}, then 
\[ ^\sigma\!\tilde U={\tiny \begin{bmatrix}
	a_{\sigma(1)} & & & b_{\sigma(1)} & & \\
	& \ddots & & & \ddots & \\
	& & a_{\sigma(n)} & & & b_{\sigma(n)}\\
	c_{\sigma(1)} & & & d_{\sigma(1)} & & \\
	& \ddots & & & \ddots & \\
	& & c_{\sigma(n)} & & & d_{\sigma(n)}\\ 
	\end{bmatrix}} \in \Sp_{2n}^{\text{loc}}(\mathbb F_2) \ . \]
It follows that $\Sp_{2n}^{\text{loc}}(\mathbb F_2)$ and $\frak S_n$ generate a subgroup of $\Sp(\mathbb F_2^{2n})$ isomorphic to the semidirect product $\Sp_{2n}^{\text{loc}}(\mathbb F_2)\rtimes_\varphi \frak S_n$
\[ \begin{matrix}
	\Sp_{2n}^{\text{loc}}(\mathbb F_2)\rtimes_\varphi \frak S_n & \stackrel{\simeq}{\longrightarrow} & \left\langle \Sp_{2n}^{\text{loc}}(\mathbb F_2) \ , \ \frak S_n\right\rangle & \subset \Sp(\mathbb F_2^{2n})\\
	(\tilde U,\sigma ) & \mapsto & \tilde U S_\sigma
	\end{matrix} \ ,\]
which acts on $\mathbb F_2^{2n}$ as in the exponents in \eqref{semidir acting on exponents}, that is
\begin{equation}\label{semidir acting on F2n}
(\tilde U,\sigma)\cdot(\mu_1, \ldots, \mu_n,\nu_1,\ldots , \nu_n) \ = \ \tilde US_\sigma(\mu_{1}, \ldots, \mu_{n},\nu_{1},\ldots , \nu_{n}) = \end{equation}
\[= \ \tilde U(\mu_{\sigma(1)}, \ldots, \mu_{\sigma(n)},\nu_{\sigma(1)},\ldots , \nu_{\sigma(n)}) \ =\]
{\footnotesize \[
		\big( a_{1}\mu_{\sigma(1)}+b_{1}\nu_{\sigma(1)}, \ \ldots, \ a_{n}\mu_{\sigma(n)}+b_{n}\nu_{\sigma(n)}, \ c_{1}\mu_{\sigma(1)}+d_{1}\nu_{\sigma(1)}, \ \ldots , \ c_{n}\mu_{\sigma(n)}+d_{n}\nu_{\sigma(n)}\big) \ . \]}
	
\indent We conclude that the restriction of $\cal C_n \rightarrow \Sp(\mathbb F_2^{2n})$ to $\cal C_n^{\text{loc}}\rtimes_\phi \frak S_n$ surjects onto $\Sp_{2n}^{\text{loc}}\rtimes_\varphi \frak S_n$ as follows
\begin{equation}\label{from cliff semidir to sympl semidir}
	\begin{matrix}
	\cal C_n^{\text{loc}} \rtimes_\phi \frak S_n &  \twoheadrightarrow  &   \Sp_{2n}^{\text{loc}}(\mathbb F_2)\rtimes_\varphi \frak S_n  \\
	(U,\sigma) & \mapsto & (\tilde U, \sigma)\\
	U\tilde\sigma  & \mapsto & \tilde U S_\sigma
	\end{matrix}
	\end{equation}
where in the last line we formally identify $(U,\sigma)$ with the Clifford transformation $U\tilde \sigma\in \cal C_n$ and $(\tilde U,\sigma)$ with the symplectic transformation $\tilde U S_\sigma \in \Sp(\mathbb F_2^{2n})$. We recall that the kernel of \eqref{from cliff semidir to sympl semidir} is $\cal P_n\rtimes_\phi \frak S_n$. \\
\indent Moreover, from equations \eqref{semidir acting on exponents} and \eqref{semidir acting on F2n} it follows that, although the two semidirect products are not isomorphic, the above homomorphism translates the action of $\cal C_n^{\text{loc}}\rtimes_\phi\frak S_n$ on the Pauli group $\cal P_n$ into the action of $\Sp_{2n}^{\text{loc}}(\mathbb F_2)\rtimes_\varphi \frak S_n$ on the symplectic space $(\mathbb F_2^{2n},J)$, and viceversa. Thus we get a correspondence between orbits (up to phases in $\cal P_n$)
\begin{equation}\label{orbits bij of P_n and F2n}
	 \faktor{\cal P_n}{ C_n^{\text{loc}}\rtimes_\phi \frak S_n} \  \longleftrightarrow \ \faktor{\mathbb F_2^{2n}}{\Sp_{2n}^{\text{loc}}(\mathbb F_2)\rtimes_\varphi \frak S_n} \ .
	 \end{equation}


\subsection{The actions $\cal C_n^{\text{loc}}\rtimes_\phi \frak S_n \curvearrowright \scr S(\cal P_n)$ and $\Sp_{2n}^{\text{loc}}(\mathbb F_2)\rtimes_\varphi \frak S_n \curvearrowright \cal I^n$}

\indent \indent The next step is to extend the orbit correspondence \eqref{orbits bij of P_n and F2n} to an orbit correspondence between the set of $n$-fold maximal abelian subgroups in $\cal P_n$
\[ \mathscr S (\cal P_n)=\left\{ \langle M_1,\ldots, M_n \rangle < \cal P_n \ | \ M_i's \  \text{independent \& mutually commuting}\right\}\]
and the set of maximal fully isotropic subspaces in $\mathbb F_2^{2n}$ (with respect to the symplectic form $J$)
\[ \cal I^n=\left\{ W \subset \mathbb P_2^{2n-1} \ | \ \langle P,Q\rangle_J=0 \ \forall P,Q \in W , \ \dim_{\mathbb P}=n-1 \right\} \ ,\]
which are in bijection via \eqref{Pauli-I^n}. From diagram \eqref{Cliff-Sp} we already have the correspondence between the orbits
\[ \faktor{\scr S(\cal P_n)}{\cal C_n^{\text{loc}}} \ \longleftrightarrow \ \faktor{\cal I^n}{\Sp_{2n}^{\text{loc}}(\mathbb F_2)} \ . \]

\paragraph{Action on $\scr S(\cal P_n)$.} Given $M_i=A_1^{(i)}\otimes \ldots \otimes A_n^{(i)}\in \cal P_n$ and $\sigma \in \frak S_n$, we denote 
	\[ ^\sigma\!M_i \ = \ \sigma \cdot (A_1^{(i)}\otimes \ldots \otimes A_n^{(i)}) \ = \ A_{\sigma(1)}^{(i)}\otimes \ldots \otimes A_{\sigma(n)}^{(i)} \ . \]
Let $S_{M_1,\ldots, M_n}=\langle M_1,\ldots, M_n\rangle \in \scr S(\cal P_n)$ be a maximal abelian subgroup of $\cal P_n$ with
	\[ M_i= \alpha_iZ^{\mu_1^{(i)}}X^{\nu_1^{(i)}}\otimes \ldots \otimes Z^{\mu_n^{(i)}}X^{\nu_n^{(i)}} \ \ \ , \ \alpha_i \in \{\pm 1, \pm i\} \ .\]
Then, for any $\sigma \in \frak S_n$, the observables $^\sigma M_i$'s are such that
	\[
	\sum_{j=1}^n\big(\mu_{\sigma(j)}^{(h)}\nu_{\sigma(j)}^{(k)}-\mu_{\sigma(j)}^{(k)}\nu_{\sigma(j)}^{(h)}\big)\stackrel{l=\sigma^{-1}(j)}{=}\sum_{l=1}^n\big(\mu_l^{(h)}\nu_l^{(k)}-\mu_l^{(k)}\nu_l^{(h)}\big)\stackrel{(\clubsuit)}{=}0
	\]
and 
\begin{align*}
	(^\sigma M_1)^{c_1}\cdot \ldots \cdot (^\sigma M_n)^{c_n} & = (A_{\sigma(1)}^{(1)})^{c_1}\cdots (A_{\sigma(1)}^{(n)})^{c_n}\otimes \ldots \otimes (A_{\sigma(n)}^{(1)})^{c_1}\cdots (A_{\sigma(n)}^{(n)})^{c_n}\\
	& \stackrel{(\spadesuit)}{=} I_2\otimes \ldots \otimes I_2 = I^{\otimes n} \ ,
	\end{align*}
where the equalities $(\clubsuit)$ and $(\spadesuit)$ respectively follow from the commutation and the independence of the $M_i$'s. It follows that $^\sigma\!M_1, \ldots, ^\sigma\!M_n$ are independent and mutually commuting too: thus we get the following well-defined action
\[ \begin{matrix} 
	\frak S_n \times \scr S(\cal P_n) & \longrightarrow & \scr S(\cal P_n)\\
	\big( \sigma \ , \ S_{M_1,\ldots, M_n} \big) & \mapsto & S_{^\sigma\!M_1, \ldots, ^\sigma\!M_n}
	\end{matrix} \ .\]
Actually, since the symmetric group $\mathfrak{S}_n$ acts on the generators of a maximal abelian subgroup, the above action coincides with the one permuting {\em both} the entries of any generators {\em and} the generators among them, that is
\begin{equation}\label{S_n on S(P_n)}
\begin{matrix} 
\frak S_n \times \scr S(\cal P_n) & \longrightarrow & \scr S(\cal P_n)\\
\big( \sigma \ , \ \langle M_1,\ldots, M_n\rangle \big) & \mapsto & \langle ^\sigma\!M_{\sigma(1)}, \ldots, ^\sigma\!M_{\sigma(n)}\rangle
\end{matrix}\end{equation}
where 
	\[ ^\sigma M_{\sigma(i)}= Z^{\mu_{\sigma(1)}^{(\sigma(i))}}X^{\nu_{\sigma(1)}^{(\sigma(i))}}\otimes \ldots \otimes Z^{\mu_{\sigma(n)}^{(\sigma(i))}}X^{\nu_{\sigma(n)}^{(\sigma(i))}} \ .\]
	
We conclude that the action of the group $\cal C_n^{\text{loc}}\rtimes_\phi\mathfrak S_n$ on $\cal P_n$ extends to the action

\begin{equation}\label{semidir on S(P_n)}
	\begin{matrix}
	\bigg( \cal C_n^{\text{loc}}\rtimes_\phi \mathfrak S_n \bigg) \times \scr S(\cal P_n) & \longrightarrow & \scr S(\cal P_n)\\
	\left( (U,\sigma) \ , \ \langle M_1,\ldots, M_n\rangle \right) & \mapsto & \langle \  U(^\sigma \!M_{\sigma(1)})U^\dagger , \ldots , U(^\sigma \!M_{\sigma(n)})U^\dagger \ \rangle
	\end{matrix}
	\end{equation}
where, for $U=U_1\otimes \ldots \otimes U_n\in \cal C_n^{\text{loc}}$ and $M_i=Z^{\mu_1^{(i)}}X^{\nu_1^{(i)}}\otimes 	\ldots \otimes Z^{\mu_n^{(i)}}X^{\nu_n^{(i)}}$, 
\[
	U(^\sigma \!M_{\sigma(i)})U^\dagger = U_1Z^{\mu_{\sigma(1)}^{(\sigma(i))}}X^{\nu_{\sigma(1)}^{(\sigma(i))}}U_1^\dagger \otimes \ldots \otimes U_nZ^{\mu_{\sigma(n)}^{(\sigma(i))}}X^{\nu_{\sigma(n)}^{(\sigma(i))}}U_n^\dagger \ . \]

\paragraph{Action on $\cal I^n$.} From the correspondence \eqref{Pauli-I^n} we know that any maximal abelian subgroup $\langle M_1,\ldots, M_n\rangle \in \scr S(P_n)$ corresponds to the maximal isotropic subspace $H_{P_1,\ldots, P_n}\in \cal I^n$, where $P_i\in \mathbb F_2^{2n}$ is defined by the exponents in $M_i$. Moreover, by Proposition \ref{subspacesascolumns} and by the linear independence of the $M_i$'s, it holds $H_{P_1,\ldots,P_n}^\vee=H_{P_1,\ldots, P_n}$, thus $H_{P_1,\ldots, P_n}=\langle P_1,\ldots, P_n\rangle_{\mathbb F_2}$. \\
By putting together the actions \eqref{semidir acting on F2n} and \eqref{semidir on S(P_n)}, one gets the action
\begin{equation}\label{semidir on I^n}
\begin{matrix}
\bigg( \Sp_{2n}^{\text{loc}}(\mathbb F_2)\rtimes_\varphi \mathfrak S_n \bigg) \times \cal I^n & \longrightarrow & \cal I^n\\
\left( (\tilde U,\sigma) \ , \ H_{P_1,\ldots , P_n} \right) & \mapsto &   H_{\tilde US_\sigma P_{\sigma(1)}, \ldots , \tilde US_\sigma P_{\sigma(n)}}
\end{matrix} \ 
\end{equation}
where, for $P_i=(\mu_{1}^{(i)}, \ldots, \mu_{n}^{(i)},\nu_{1}^{(i)},\ldots , \nu_{n}^{(i)})$,
\begin{equation}\label{semidir on I^n - generators}
\tilde US_\sigma P_{\sigma(i)} \ = \ \tilde US_\sigma\left(\mu_{1}^{(\sigma(i))}, \ldots, \mu_{n}^{(\sigma(i))},\nu_{1}^{(\sigma(i))},\ldots , \nu_{n}^{(\sigma(i))}\right) = \end{equation}
\[= \ \tilde U\left(\mu_{\sigma(1)}^{(\sigma(i))}, \ldots, \mu_{\sigma(n)}^{(\sigma(i))},\nu_{\sigma(1)}^{(\sigma(i))},\ldots , \nu_{\sigma(n)}^{(\sigma(i))}\right) \ =\]
{\footnotesize \[
	\big( a_{1}\mu_{\sigma(1)}^{(\sigma(i))}+b_{1}\nu_{\sigma(1)}^{(\sigma(i))}, \ \ldots, \ a_{n}\mu_{\sigma(n)}^{(\sigma(i))}+b_{n}\nu_{\sigma(n)}^{(\sigma(i))}, \ c_{1}\mu_{\sigma(1)}^{(\sigma(i))}+d_{1}\nu_{\sigma(1)}^{(\sigma(i))}, \ \ldots , \ c_{n}\mu_{\sigma(n)}^{(\sigma(i))}+d_{n}\nu_{\sigma(n)}^{(\sigma(i))}\big) \ . \]} 

\noindent Notice that each point $\tilde US_\sigma P_{\sigma(i)}\in \mathbb F_2^{2n}$ corresponds to the observable $U(^\sigma\!M_{\sigma(i)})U^\dagger \in \cal P_n$. We conclude that the correspondence \eqref{orbits bij of P_n and F2n} extends to a correspondence between the orbits
\begin{equation}\label{orbits bij of S(P_n) and I^n}
\faktor{\scr S(\cal P_n)}{ C_n^{\text{loc}}\rtimes_\phi \frak S_n} \  \longleftrightarrow \ \faktor{\cal I^n}{\Sp_{2n}^{\text{loc}}(\mathbb F_2)\rtimes_\varphi \frak S_n} \ .
\end{equation}

\subsection{The actions of $\Sp_{2n}^{\text{loc}}(\mathbb F_2)\rtimes_\varphi \mathfrak S_n$ on $\LG_{\mathbb F_2}(n,2n)$ and $\cal Z_n$}

\indent \indent We are close to prove the correspondence between orbits in $\scr S(\cal P_n)$ with respect to $\cal C_n^{\text{loc}}\rtimes_\phi \mathfrak S_n$, and orbits in the variety of binary symmetric principal minors $\cal Z_n$ with respect to $\Sp_{2n}^{\text{loc}}(\mathbb F_2)\rtimes_\varphi \mathfrak S_n$: it only remains to translate the action \eqref{semidir on I^n} into an action on the Lagrangian Grassmannian $\LG_{\mathbb F_2}(n,2n)$, and later on $\cal Z_n$ via the diagram \eqref{finalobjects}.

\paragraph{Action on $\LG_{\mathbb F_2}(n,2n)$.} By definition, $\LG_{\mathbb F_2}(n,2n)=\Pl(\cal I^n) \subset \mathbb P\big(\bigwedge^{n}\mathbb F_2^{2n} \big)$, where $\Pl$ is the Pl\"{u}cker embedding: in particular, a maximal fully isotropic subspace $H_{P_1,\ldots, P_n}=\langle P_1,\ldots, P_n\rangle_{\mathbb F_2}\in \cal I^n$ corresponds to the point $[P_1\wedge \ldots \wedge P_n]\in \LG_{\mathbb F_2}(n,2n)$. Moreover, one can write the point $[P_1\wedge \ldots \wedge P_n]$ in coordinates in $\mathbb P_2^{\binom{2n}{n}-1}$: given $N=\big[ P_1 | \ldots | P_n\big]$ the $2n \times n$ matrix representing the subspace $H_{P_1,\ldots, P_n}$, the coordinates of $[P_1\wedge \ldots \wedge P_n]$ are given by the maximal minors (i.e. of size $n\times n$) of $N$, that is
\[ \begin{matrix}
	\mathbb P\left(\bigwedge^n\mathbb F_2^{2n}\right) & \longleftrightarrow & \mathbb P_2^{\binom{2n}{n}-1}\\
	[P_1\wedge \ldots \wedge P_n] & \longleftrightarrow & \left[N_{\{1,\ldots,n\}}: N_{\{1,\ldots, n-1,n+1\}}: \ldots : N_{\{n+1,\ldots, 2n\}}\right]	
	\end{matrix}\]
where $N_I$ is the minor of $N$ given by the $I$-indexed rows and all the $n$ columns.\\
\indent It is straightforward that the action \eqref{semidir on I^n} of $\Sp_{2n}^{\text{loc}}(\mathbb F_2)\rtimes_\varphi \mathfrak S_n$ on $\cal I^n$ is equivalent to the following action on the Lagrangian Grassmannian (which extends the action \eqref{Sp on LG}):
\begin{equation}\label{semidir on LG}
\begin{matrix}
\bigg( \Sp_{2n}^{\text{loc}}(\mathbb F_2)\rtimes_\varphi \mathfrak S_n \bigg) \times \LG_{\mathbb F_2}(n,2n) & \longrightarrow & \LG_{\mathbb F_2}(n,2n)\\
\left( (\tilde U,\sigma) \ , \ [P_1\wedge \ldots \wedge P_n] \right) & \mapsto &   \left[ \tilde US_\sigma P_{\sigma(1)} \wedge  \ldots \wedge \tilde US_\sigma P_{\sigma(n)}\right]
\end{matrix}
\end{equation}
where $\tilde US_\sigma P_{\sigma(i)}$ are as in \eqref{semidir on I^n - generators}. 

\begin{obs}
	From Remark \ref{minors of A'} we know that the action of $\Sp_{2n}^{\text{loc}}(\mathbb F_2)$ does not preserve the open subsets $LU_I$ defined in \eqref{charts in LG}, hence the above action does not preserve them either. 
	\end{obs}

\indent We can translate the action \eqref{semidir on LG} on the Lagrangian Grassmannian into an action on the set of full-rank $2n\times n$ matrices. By \eqref{Sp-Pl-matrix} we already know that $\Sp_{2n}^{\text{loc}}(\mathbb F_2)$ acts on the full-rank $2n\times n$ matrices by left-multiplication, that is a transformation $\tilde U={\scriptsize \begin{bmatrix} A& B \\ C & D \end{bmatrix}}\in \Sp_{2n}^{\text{loc}}(\mathbb F_2)$ maps a certain full-rank $2n\times n$ matrix $N={\scriptsize \begin{bmatrix} F \\G \end{bmatrix}}$ into the full-rank $2n\times n$ matrix
\begin{equation}\label{Sp on matrices} 
	\tilde U \cdot N \ = \ \begin{bmatrix}
	AF+BG \\CF+DG
	\end{bmatrix} \ . 
	\end{equation}
\indent By substituting the identity matrix $\tilde U=I$ in \eqref{semidir on I^n - generators} we deduce that a permutation $\sigma \in \mathfrak S_n$ acts on a full-rank $2n\times n$ matrix ${\scriptsize \begin{bmatrix} F\\ G\end{bmatrix}}=[P_1| \ldots |P_n]$ as
\begin{equation}\label{Sn on matrices}
	\sigma \cdot \begin{bmatrix} F\\ G\end{bmatrix}= \left[S_\sigma P_{\sigma(1)}| \ldots | S_\sigma P_{\sigma(n)}\right] = {\tiny \begin{bmatrix}
	\begin{matrix}
	\mu_{\sigma(1)}^{(\sigma(1))} & \mu_{\sigma(1)}^{(\sigma(2))} & \cdots & \mu_{\sigma(1)}^{(\sigma(n))}\\
	\mu_{\sigma(2)}^{(\sigma(1))} & \mu_{\sigma(2)}^{(\sigma(2))} & \cdots & \mu_{\sigma(2)}^{(\sigma(n))}\\
	\vdots & \vdots & & \vdots\\
	\mu_{\sigma(n)}^{(\sigma(1))} & \mu_{\sigma(n)}^{(\sigma(2))} & \cdots & \mu_{\sigma(n)}^{(\sigma(n))}
	\end{matrix}\\
	---------------\\
	\begin{matrix}
	\nu_{\sigma(1)}^{(\sigma(1))} & \nu_{\sigma(1)}^{(\sigma(2))} & \cdots & \nu_{\sigma(1)}^{(\sigma(n))}\\
	\nu_{\sigma(2)}^{(\sigma(1))} & \nu_{\sigma(2)}^{(\sigma(2))} & \cdots & \nu_{\sigma(2)}^{(\sigma(n))}\\
	\vdots & \vdots & & \vdots\\
	\nu_{\sigma(n)}^{(\sigma(1))} & \nu_{\sigma(n)}^{(\sigma(2))} & \cdots & \nu_{\sigma(n)}^{(\sigma(n))}
	\end{matrix}
	\end{bmatrix}} = \begin{bmatrix} ^\sigma\!F \\ ^\sigma\!G \end{bmatrix}
	\end{equation}
where the $n\times n$ matrix $^\sigma\!F$ (resp. $^\sigma\!G$) is obtained by the $n \times n$ matrix $F$ (resp. $G$) by permuting both columns and rows by $\sigma \in \mathfrak S_n$: more precisely, if $S_\sigma = {\scriptsize \begin{bmatrix} A_\sigma \\
	& A_\sigma \end{bmatrix}}$ where $A_\sigma$ is the $n\times n$ permutation matrix defined by $\sigma$, then the action of $\sigma$ onto ${\scriptsize \begin{bmatrix} F\\ G\end{bmatrix}}$ corresponds to the conjugacy action by $A_\sigma$ onto $F$ and $G$ separately, that is 
\[ \sigma \cdot \begin{bmatrix} F\\G \end{bmatrix}= \begin{bmatrix}^\sigma\!F\\^\sigma\!G \end{bmatrix}= \begin{bmatrix} A_\sigma F A_\sigma^{-1}\\A_\sigma G A_\sigma^{-1}\end{bmatrix} \ .\]

\begin{obs}\label{not preserve LU}
	By \eqref{Sn on matrices} the above action preserves the full-rankness. Moreover, from the matrix setting we deduce that, analogously to Remark \eqref{minors of A'}, the action \eqref{semidir on LG} of $\mathfrak S_n$ on $\LG_{\mathbb F_2}(n,2n)$ does not preserve the open subsets $LU_I$ either. For instance, denote by $E_{ij}$ the $n\times n$ matrix having $1$ in the entry $(i,j)$ and zero everywhere else: then ${\scriptsize \begin{bmatrix} E_{ii}+E_{jj}\\I_n \end{bmatrix}}$ {\em lies} in the open subset $LU_{\{i,j,n+1,\ldots, 2n\}\setminus \{n+i,n+j\}}$ but the permutation $\sigma=(ik)(j\ell)$ maps it into the matrix ${\scriptsize \begin{bmatrix} E_{kk}+E_{\ell\ell}\\I_n \end{bmatrix}}$ which does not lie in $LU_{\{i,j,n+1,\ldots, 2n\}\setminus \{n+i,n+j\}}$. However, the action of $\mathfrak S_n$ preserves the open subset $LU_{\{1,\ldots,n\}}$: indeed, $\sigma\cdot {\scriptsize \begin{bmatrix}I_n\\S \end{bmatrix}}={\scriptsize \begin{bmatrix}^\sigma\!I_n\\^\sigma\!S \end{bmatrix}}={\scriptsize \begin{bmatrix}I_n\\^\sigma\!S \end{bmatrix}}$.
	\end{obs}

By putting together \eqref{Sp on matrices} and \eqref{Sn on matrices} we conclude that the action \eqref{semidir on LG} of $\Sp_{2n}^{\text{loc}}\rtimes_\varphi \mathfrak S_n$ on $\LG_{\mathbb F_2}(n,2n)$ is equivalent to the restriction onto full-rank matrices of the action
\begin{equation}\label{semidir on matrices}
\begin{matrix}
\bigg( \Sp_{2n}^{\text{loc}}(\mathbb F_2)\rtimes_\varphi \mathfrak S_n \bigg) \times \Mat_{2n\times n}(\mathbb F_2) & \longrightarrow & \Mat_{2n\times n}(\mathbb F_2)\\
\left( (\tilde U,\sigma) \ , \ {\footnotesize \begin{bmatrix} F\\G \end{bmatrix}} \right) & \mapsto &   {\footnotesize \begin{bmatrix} A(^\sigma\!F)+B(^\sigma\!G)\\
	C(^\sigma\!F)+D(^\sigma\!G)  \end{bmatrix}}
\end{matrix}
\end{equation}
where $\tilde U={\scriptsize \begin{bmatrix} A& B \\ C & D \end{bmatrix}}$.

\paragraph{Action on $\cal Z_n$.} By Proposition \ref{LGisoZ} the Lagrangian Grassmannian $\LG_{\mathbb F_2}(n,2n)$ and the variety of binary symmetric principal minors $\cal Z_n$ are in bijection, thus we can easily conclude that the action \eqref{semidir on LG} of $\Sp_{2n}^{\text{loc}}(\mathbb F_2)\rtimes_\varphi \mathfrak S_n$ on $\LG_{\mathbb F_2}(n,2n)$ induces an action on $\cal Z_n$ making the following diagram commutative:
\begin{equation}\label{semidir on Zn}
\begin{tikzpicture}[scale=2.5]

\node(LG1) at (-1,0.3){$\bigg( \Sp_{2n}^{\text{loc}}(\mathbb F_2)\rtimes_\varphi \mathfrak S_n \bigg) \times \LG_{\mathbb F_2}(n,2n)$};
\node(LG2) at (1,0.3){$\LG_{\mathbb F_2}(n,2n)$};
\node(Z1) at (-1,-0.3){$\bigg( \Sp_{2n}^{\text{loc}}(\mathbb F_2)\rtimes_\varphi \mathfrak S_n \bigg) \times \cal Z_n$};
\node(Z2) at (1,-0.3){$\cal Z_n$};
\node(circ) at (0,0){$\circlearrowright$};

\path[font=\scriptsize,->, >= angle 90]

(LG1) edge [->] node [above] {} (LG2)
(LG1) edge [->] node [above] {} (Z1)
(LG2) edge [->] node [below] {} (Z2)
(Z1) edge [->] node [above] {} (Z2);

\end{tikzpicture}
 \ .\end{equation}

\indent In the following we show that this induced action on $\cal Z_n$ actually is a natural one.\\
Fix $(v_0,v_1)$ a basis of $\mathbb F_2^2$ and the induced basis $\big(|i_1\ldots i_n\rangle = v_{i_1}\otimes \ldots \otimes v_{i_n} \ | \ i_k\in \{0,1\}\big)$ (in Dirac notation) of $\mathbb F_2^2 \otimes \ldots \otimes \mathbb F_2^2$. Consider the isomorphism 

\begin{equation}\label{iso Dirac}
	\begin{matrix}
		\mathbb F_2^2 \otimes \ldots \otimes \mathbb F_2^2 &\stackrel{\simeq}{\longrightarrow} & \mathbb F_2^{2^n}\\
		|0\ldots 0\rangle & \mapsto & (1, 0 , \ldots ,0)\\
		|0\ldots \underbrace{1}_{n-i}\ldots 0 \rangle & \mapsto &  (0,\ldots , \underbrace{1}_{i+2}, \ldots, 0)\\
		\vdots\\
		|1\ldots 1\rangle &\mapsto & (0,\ldots ,0,1)
	\end{matrix}
\end{equation}

\noindent By Proposition \ref{Zn invariant} we know that the natural action of $\SL(2,\mathbb F_2)^{\times n}$ on $\mathbb P\big(\mathbb F_2^{2^n}\big)$
\[(U_1,\ldots , U_n) \cdot [v_{1}\otimes \ldots \otimes v_n] = \big[U_1v_1 \otimes \ldots \otimes U_nv_n\big]\]

\noindent induces an action of $\SL(2,\mathbb F_2)^{\times n}$ on $\cal Z_n\subset \mathbb P_2^{2^n-1}$. Another natural action on $\mathbb P\big(\mathbb F_2^{2^n}\big)$ is the one of the symmetric group $\frak S_n$ given by
\[\sigma \cdot [v_1\otimes \ldots \otimes v_n]= [v_{\sigma(1)}\otimes \ldots \otimes v_{\sigma(n)}]\]
\noindent which permutes the tensor product entries: it is known that this action preserves $\cal Z_n$. 

\begin{obs}
	Via the isomorphism \eqref{iso Dirac}, one can check that the action of $\mathfrak S_n$ preserves the standard open subset $\cal ZU_{\{1,\ldots, n\}}\subset \cal Z_n$ since
\begin{equation}\label{Sn on Zn}
	\sigma \cdot \big[1:s_{ii} : S_{[i,j]} : \ldots : \det S \big] = \big[1: s_{\sigma(i)\sigma(i)} : S_{[\sigma(i),\sigma(j)]} : \ldots : \det S\big] \ .
\end{equation}
This agrees with the end of Remark \ref{not preserve LU}. But, in general, this action does not preserve the open subsets $LU_I$.
\end{obs} 

\indent Moreover, the action of $\frak S_n$ on $\mathbb P\big(\mathbb F_2^2 \otimes \ldots \otimes \mathbb F_2^2\big)$ permuting the tensor entries induces an action of $\frak S_n$ on $\SL(2,\mathbb F_2)^{\times n}$ permuting the entries of the direct product (and the latter corresponds to the action \eqref{S_n acting on Sp^loc} via the isomorphism $\SL(2,\mathbb F_2)^{\times n}\simeq \Sp_{2n}^{\text{loc}}(\mathbb F_2)$). We conclude that there is a natural action 
\[ \bigg( \SL(2,\mathbb F_2)^{\times n}\rtimes \frak S_n \bigg) \times \cal Z_n \ \longrightarrow \ \cal Z_n \]
{\em and} that it actually corresponds, via the isomorphism $\SL(2,\mathbb F_2)^{\times n}\simeq \Sp_{2n}^{\text{loc}}(\mathbb F_2)$, to the action on $\cal Z_n$ in \eqref{semidir on Zn}.

\paragraph{Conclusion.} This section and Section \ref{orbitZn} achieve the proof of the first part of Theorem \ref{theo:main} by establishing the bijection between the $\left(\mathcal{C}^{\text{loc}} _n\rtimes \mathfrak{S}_n\right)$-orbits of maximal abelian subgroups of $\mathcal{P}_n$, or equivalently maximal fully isotropic subspaces of $\mathcal{W}(2n-1,2)$, and the $\left(\SL(2,\mathbb{F}_2)^{\times n}\rtimes \mathfrak{S}_n\right)$-orbits of $\mathcal{Z}_n\subset \mathbb{P}(\mathbb{F}_2^{2^n})$:
\begin{equation}\label{orbits bij of S(P_n) and Zn}
\faktor{\scr S(\cal P_n)}{ C_n^{\text{loc}}\rtimes_\phi \frak S_n} \  \longleftrightarrow \ \faktor{\cal Z_n}{\Sp_{2n}^{\text{loc}}(\mathbb F_2)\rtimes_\varphi \mathfrak S_n} \ \simeq \ \faktor{\cal Z_n}{\SL(2,\mathbb F_2)^{\times n}\rtimes \mathfrak S_n} \ .
\end{equation}

\section{Stabilizer states and their orbits under  $\cal C_n^{\text{loc}}\rtimes \mathfrak S_n$}\label{sec:stab}

\indent \indent In this section we focus on a subset of the set $\scr S(\cal P_n)$ of the maximal ($n$-fold) abelian subgroups of the Pauli group $\cal P_n$, and show that the action \eqref{semidir on S(P_n)} restricts to an action on this subset.

\begin{Def}
	 A subgroup $S<\cal P_n$ is called {\em stabilizer group} if it is abelian and $-I^{\otimes n}\notin S$. In particular, $S\in \scr S(\cal P_n)$ is a {\bf maximal stabilizer group} if it does not contain $-I^{\otimes n}$. We denote the set of maximal stabilizer state groups in $\cal P_n$ by $\scr S^+(\cal P_n)$, or simply $\scr S^+_n$. 
	\end{Def}

Let $S=\langle M_1,\ldots, M_n\rangle \in \scr S_n^+$: since the $M_i$'s are mutually commuting, they admit (at least) one common eigenvector $|\phi\rangle\in (\mathbb C^2)^{\otimes n}$. Moreover, recall that for any $M\in \cal P_n$ it holds $M^2=\pm I^{\otimes n}$: in particular, for any $M \in S$ one gets $M^2=I^{\otimes n}$ since $-I^{\otimes n}\notin S$ and $M^2 \in S$. It follows that elements in $S$ can only have eigenvalues $+ 1$ or $-1$.\\
\indent For any abelian subgroup $S<\cal P_n$ not containing $-I^{\otimes n}$ the set of the common eigenvectors with eigenvalue $+1$ of $S$ forms a subspace $V_S$, called {\em stabilized subspace} or {\em stabilizer code} of $S$, and its dimension is $\dim V_S=2^n/|S|$  \cite[Sec. III.B, III.C]{dangniam2020optimal}: in particular, if $S \in \scr S_n^+$, $V_S$ is one-dimensional.

\begin{Def}
	A ($n$-qubit) {\bf stabilizer state} is the (unique up to phase) common eigenvector with eigenvalue $+1$ of a maximal stabilizer state $S \in \scr S_n^+$. We denote by $\Phi_1^n$ the set of $n$-qubit stabilizer states.
	\end{Def}

\noindent There is a one-to-one correspondence between $\Phi_1^n$ and $\scr S_n^+$: we denote by $S_{|\phi\rangle}$ the maximal stabilizer group associated to the stabilizer state $|\phi\rangle$. Let us study how $\cal C_n^{\text{loc}}\rtimes \frak S_n$ acts on $\scr S_n^+$ and $\Phi_1^n$.\\
\hfil\break
\indent First of all, the action \eqref{C_n^loc on S(P_n)} of $\cal C_n^{\text{loc}}$ on $\scr S(\cal P_n)$ preserves $\scr S_n^+$: indeed, given $U \in \cal C_n^{\text{loc}}$ and $S_{|\phi\rangle}=\langle M_1, \ldots , M_n\rangle\in \scr S_n^+$, it holds \[UM_iU^\dagger (U|\phi\rangle) = UM_i|\phi\rangle =U|\phi\rangle \ \ \ ,\forall i=1:n\]
that is $U|\phi\rangle$ is common eigenvector with eigenvalue $+1$ of $S_{UM_1U^\dagger, \ldots , UM_nU^\dagger}$.\\
In particular, the action of $\cal C_n^{\text{loc}}$ on $\scr S_n^+$ is given by

\begin{equation}\label{Cn on true stab}
\begin{matrix}
\cal C_n^{\text{loc}} \ \times \ \scr S_n^+ & \longrightarrow & \scr S_n^+\\
(U\ ,\  S_{|\phi\rangle}) & \mapsto & S_{U|\phi\rangle}
\end{matrix}
\end{equation}

\noindent and it is equivalent to the action of $\cal C_n^{\text{loc}}$ on $\Phi_1^n$ given by

\begin{equation}\label{Cn on stabstates}
\begin{matrix}
\cal C_n^{\text{loc}} \ \times \ \Phi_1^n & \longrightarrow & \Phi_1^n\\
(U\ ,\  |\phi\rangle) & \mapsto & U|\phi\rangle
\end{matrix}
\end{equation}

\noindent where, if $U=U_1\otimes \ldots \otimes U_n$ and $|\phi\rangle =\sum_{j}|\phi_j^{(1)}\rangle \otimes \ldots \otimes |\phi_j^{(n)}\rangle$ (with $|\phi_j^{(k)}\rangle \in \mathbb C^2$), 
\[ U|\phi\rangle = \sum_{j}U_1|\phi_j^{(1)}\rangle \otimes \ldots \otimes U_n|\phi_j^{(n)}\rangle \ .\]

\indent On the other hand, also the action \eqref{S_n on S(P_n)} of $\frak S_n$ on $\scr S(\cal P_n)$ preserves $\scr S_n^+$: indeed, given $S=\langle M_1,\ldots, M_n\rangle \in \scr S_n^+$ (which by definition does not contain $-I^{\otimes n}$), if $^\sigma\!S=\langle ^\sigma M_{\sigma(1)},\ldots, ^\sigma\!M_{\sigma(n)}\rangle$ contained $-I^{\otimes n}$, then $\sigma^{-1}\cdot (-I^{\otimes n})=-I^{\otimes n}$ would be in $^{\sigma^{-1}}(^\sigma\!S)=S$ which is a contradiction. Thus we get the induced action 
\[\begin{matrix}
\frak S_n \ \times \ \scr S_n^+ & \longrightarrow & \scr S_n^+\\
\big(\sigma \ , \ S_{|\phi\rangle}\big) & \mapsto & ^\sigma\!(S_{|\phi\rangle})
\end{matrix} \ .\]
Actually, we can exhibit the stabilizer state corresponding to $^\sigma\!(S_{|\phi\rangle})$: given $\sigma \in \frak S_n$ and $|\phi\rangle =\sum_{j}|\phi_j^{(1)}\rangle \otimes \ldots \otimes |\phi_j^{(n)}\rangle \in \Phi_1^n$, our candidate is
\begin{equation}\label{Sn on stab state} ^\sigma\!|\phi\rangle = \sum_{j}|\phi_j^{(\sigma(1))}\rangle \otimes \ldots \otimes |\phi_j^{(\sigma(n))}\rangle \in (\mathbb C^2)^{\otimes n}\end{equation}
that is the element obtained by permuting via $\sigma$ the tensor entries of $|\phi\rangle$. Notice that this is precisely the one we should expect: indeed, $\frak S_n$ acts on a maximal stabilizer group by permuting both its generators and their tensor entries, but the eigenvector is invariant under reordering of the generators, so the only non-trivial action is the one permuting the tensor entries.

\begin{prop}
	For any stabilizer state $|\phi\rangle \in \Phi_1^n$ and any permutation $\sigma \in \mathfrak S_n$, the tensor element $^\sigma\!|\phi\rangle$ in \eqref{Sn on stab state} is also a stabilizer state, i.e. $^\sigma\!|\phi\rangle \in \Phi_1^n$. In particular, \[^\sigma\!(S_{|\phi\rangle})=S_{^\sigma\!|\phi\rangle} \in \scr S_n^+ \ . \]
	\end{prop}
\begin{proof}
	Let $|\phi\rangle=\sum_j |\phi_j^{(1)}\rangle \otimes \ldots \otimes |\phi_j^{(n)}\rangle \in \Phi_1^n$ be a stabilizer state and $S_{|\phi\rangle}=\langle M_1,\ldots, M_n\rangle \in \scr S_n^+$ the corresponding maximal stabilizer group. If $M_i=A_1^{(i)}\otimes \ldots \otimes A_n^{(n)}$ for any $i=1:n$, then by definition
	{\small \begin{align*} 
		\sum_j A_1^{(i)}|\phi_j^{(1)}\rangle \otimes \ldots \otimes A_n^{(i)}|\phi_j^{(n)}\rangle &	= (A_1^{(i)}\otimes \ldots \otimes A_n^{(i)})|\phi\rangle = M_i|\phi\rangle = |\phi\rangle \\ & = \sum_j |\phi_j^{(1)}\rangle \otimes \ldots \otimes |\phi_j^{(n)}\rangle	
		\end{align*}}
	for any $i=1:n$. By applying $\sigma \in \mathfrak S_n$ we get for any $i=1:n$
	{\small \begin{align*}
		\sum_j A_{\sigma(1)}^{(i)}|\phi_j^{(\sigma(1))}\rangle \otimes \ldots \otimes A_{\sigma(n)}^{(i)}|\phi_j^{(\sigma(n))}\rangle &
		\stackrel{(\clubsuit)}{=} \sum_j \sigma\cdot \left(A_1^{(i)}|\phi_j^{(1)}\rangle \otimes \ldots \otimes A_n^{(i)}|\phi_j^{(n)}\rangle\right)\\
		& = \sigma \cdot \left(\sum_j A_1^{(i)}|\phi_j^{(1)}\rangle \otimes \ldots \otimes A_n^{(i)}|\phi_j^{(n)}\rangle \right)\\
		& = \sigma\cdot \left( \sum_j |\phi_j^{(1)}\rangle \otimes \ldots \otimes |\phi_j^{(n)}\rangle \right)\\
		& = \sum_j |\phi_j^{(\sigma(1))}\rangle \otimes \ldots \otimes |\phi_j^{(\sigma(n))}\rangle \ ,
		\end{align*}}
	where $(\clubsuit)$ follows by definition of the action of $\mathfrak S_n$ on $(\mathbb C^2)^{\otimes n}$. By \eqref{Sn on stab state} the above chain of equalities gives 
	\[^\sigma\!M_i( ^\sigma\!|\phi\rangle)= \left(A_{\sigma(1)}^{(i)})\otimes \ldots \otimes  A_{\sigma(n)}^{(i)}\right)(^\sigma\!|\phi\rangle) \ = \ ^\sigma\!|\phi\rangle \ , \]
	that is $^\sigma\!|\phi\rangle$ is eigenvector with eigenvalue $+1$ for $^\sigma \!M_i=A_{\sigma(1)}^{(i)}\otimes \ldots \otimes  A_{\sigma(n)}^{(i)}$. Since this holds for any $i=1:n$, then $^\sigma\!|\phi\rangle$ is a common eigenvector with eigenvalue $+1$ of $\langle ^\sigma\!M_1,\ldots, ^\sigma\!M_n\rangle= \langle ^\sigma\!M_{\sigma(1)},\ldots,  ^\sigma\!M_{\sigma(n)}\rangle = \ ^\sigma\!(S_{|\phi\rangle})$. But $^\sigma\!(S_{|\phi\rangle})$ is a maximal stabilizer group, thus $^\sigma|\phi\rangle$ actually is its {\em unique} eigenvector with eigenvalue $+1$, that is $^\sigma\!|\phi\rangle \in \Phi_1^n$ and $^\sigma\!(S_{|\phi\rangle})=S_{^\sigma\!|\phi\rangle}$.
	\end{proof}

\begin{cor}
	The action \eqref{S_n on S(P_n)} of $\frak S_n$ on $\scr S(\cal P_n)$ preserves $\scr S_n^+$ and induces the action
	\begin{equation}\label{Sn on true stab group}
	\begin{matrix}
	\frak S_n \ \times \ \scr S_n^+ & \longrightarrow & \scr S_n^+\\
	\big(\sigma \ , \ S_{|\phi\rangle}\big) & \mapsto & ^\sigma\!(S_{|\phi\rangle})=S_{^\sigma\!|\phi\rangle}
	\end{matrix} \ .
	\end{equation}
	Moreover, the above action corresponds to an action on the stabilizer states
	\begin{equation}\label{Sn on true stab state}
	\begin{matrix}
	\frak S_n \ \times \ \Phi_1^n & \longrightarrow & \Phi_1^n\\
	\big(\sigma \ , \ |\phi\rangle\big) & \mapsto & ^\sigma\!|\phi\rangle
	\end{matrix} \ .
	\end{equation}
	\end{cor}

\noindent By putting together the previous actions we get that the action \eqref{semidir on S(P_n)} of $\cal C_n^{\text{loc}}\rtimes \frak S_n$ on $\scr S(\cal P_n)$ restricts to the action 
\begin{equation}\label{semidir on stab group}
\begin{matrix}
\big(\cal C_n^{\text{loc}}\rtimes \frak S_n\big) \ \times \scr S_n^+ & \longrightarrow & \scr S_n^+\\
\big((U,\sigma) \ , \ S_{|\phi\rangle}\big) & \mapsto & S_{U(^\sigma\!|\phi\rangle)}=U(^\sigma\!S_{|\phi\rangle})U^\dagger
\end{matrix}
\end{equation}
\noindent In particular, this action corresponds to the action on the stabilizer states
\begin{equation}\label{semidir on stab state}
\begin{matrix}
\big(\cal C_n^{\text{loc}}\rtimes \frak S_n\big) \times \ \Phi_1^n & \longrightarrow & \Phi_1^n\\
\big( (U,\sigma) \ ,\ |\phi\rangle \big) & \mapsto & U(^\sigma\!|\phi\rangle)
\end{matrix}
\end{equation}
where, if $U=U_1\otimes \ldots \otimes U_n$ and $|\phi\rangle =\sum_{j}|\phi_j^{(1)}\rangle \otimes \ldots \otimes |\phi_j^{(n)}\rangle$, then
\[ U(^\sigma\!|\phi\rangle) = \sum_{j}U_{1}|\phi_j^{(\sigma(1))}\rangle \otimes \ldots \otimes U_{n}|\phi_j^{(\sigma(n))}\rangle \ .\]

\paragraph{Conclusion.} This section proves an intermediate step to the second statement of Theorem \ref{theo:main}. From Section 4 we have the one-to-one correspondence \eqref{orbits bij of S(P_n) and Zn} between the orbits of maximal abelian subgroups under $\cal C_n^{\text{loc}}\rtimes \mathfrak S_n$ and the orbits in $\cal Z_n$ under $\SL(2,\mathbb F_2)^{\times n}\rtimes \mathfrak S_n$. But the Lagrangian mapping (Sec.\ref{preliminaries}) associates maximal fully isotropic subspaces in $\cal I^n$ to maximal abelian subgroups in $\scr S(\cal P_n)$ {\em up to phasis} (since we work into the quotient $V_n=\cal P_n/Z(\cal P_n)$): this means that the Lagrangian mapping does not distinguish between maximal abelian subgroups containing $-I^{\otimes n}$ and maximal stabilizer states. It follows that there is a one-to-one correspondence between the orbits of stabilizer states under $\cal C_n^{\text{loc}}\rtimes \mathfrak S_n$ and the orbits in $\cal Z_n$ under $\SL(2,\mathbb F_2)^{\times n}\rtimes \mathfrak S_n$:
\begin{equation}\label{orbits bij of stab and Zn}
\faktor{\Phi_1^n}{ C_n^{\text{loc}}\rtimes_\phi \mathfrak S_n} \  \longleftrightarrow \ \faktor{\scr S_n^+}{ C_n^{\text{loc}}\rtimes_\phi \mathfrak S_n} \ \longleftrightarrow \ \faktor{\cal Z_n}{\SL(2,\mathbb F_2)^{\times n}\rtimes \mathfrak S_n} \ .
\end{equation}

\section{Graph states and their orbits under $\cal C_n^{\text{loc}}\rtimes \mathfrak S_n$}\label{graph}

\indent \indent In this section we investigate how the action \eqref{semidir on stab state} of $\cal C_n^{\text{loc}}\rtimes \mathfrak S_n$ on $\Phi_1^n$ behaves with respect to a privileged subset of stabilizer states, the so-called {\em graph states}. 

\begin{obs}\label{rmk improper orbits}
	In \cite{article1} it is pointed out that the action of $\cal C_n^{\text{loc}}$ does not preserve this subset, hence the action \eqref{semidir on stab state} does not either: in particular, one cannot talk about {\em orbits of graph states under $\cal C_n^{\text{loc}}\rtimes \mathfrak S_n$}. However, we will improperly refer to \textquotedblleft orbit\textquotedblright \space of a graph state as to the set of all graph states belonging to the same orbit in $\Phi_1^n$: in order to get an actual action on the graph states, one should consider {\em only} certain local Clifford operations corresponding to graph transformations (cf. \cite[Sec. IV, Definition 1]{article1}).
	\end{obs}

\indent Consider a (non-oriented) graph $\Gamma=(V,E)$ defined by a finite set of vertices $V=\{1,\ldots , n\}$ and a set of edges $E\subset V\times V$. We can associate to $\Gamma$ a unique symmetric matrix $\theta \in \Sym^2(\mathbb F_2^n)$ such that
\[ \theta_{ij}=1 \iff (i,j) \in E\]
(the matrix is indeed symmetric since the graph is not oriented). The graph $\Gamma$ (or equivalently its matrix $\theta$) defines $n$ elements in the Pauli group $\cal P_n$
\[
\begin{matrix}
M_1= & Z^{\theta_{11}}X & \otimes & Z^{\theta_{12}} & \otimes & \ldots & \otimes & Z^{\theta_{1n}}\\
 \vdots\\
M_n= & Z^{\theta_{1n}} & \otimes & \ldots & \otimes & Z^{\theta_{n-1,n}} & \otimes & Z^{\theta_{nn}}X
\end{matrix} \ ,\]
or equivalently, in a more compact notation, 
\begin{equation}\label{graph observables}
 	M_i=Z^{\theta_{i1}}X^{\delta_{1i}}\otimes \ldots \otimes Z^{\theta_{in}}X^{\delta_{ni}} \ \ \ , \ \forall i=1:n \ , 
 	\end{equation}
where $\delta_{ij}$ is the Kronecker symbol: these elements are independent and mutually commuting, thus they generate a maximal abelian subgroup which we denote by $S_\Gamma=S_\theta \in \scr S(\cal P_n)$. Moreover, the maximal fully isotropic subspace $H_\theta \in \cal I^n$ corresponding to $S_\theta$ is described (in the coordinates \eqref{articlecoordinates}) by the $2n \times n$ matrix ${\scriptsize \begin{bmatrix}
	\theta \\
	I_n
	\end{bmatrix}}$ (it is already in a Pl\"{u}cker form) and belongs to the chart $LU_{\{n+1,\ldots, 2n\}}\subset \LG_{\mathbb F_2}(n,2n)$.

\begin{obs}\label{other graph chart}
	The association $\Gamma \mapsto S_\theta$ is not unique: indeed, instead of defining the elements $M_i$'s in \eqref{graph observables}, one could choose to associate to $\Gamma$ the elements
	\[ M_i'=Z^{\delta_{1i}}X^{\theta_{i1}}\otimes \ldots \otimes Z^{\delta_{ni}}X^{\theta_{in}} \ \]
	defining the subspace $H_\theta'\in LU_{\{1,\ldots ,n\}}$ described by the matrix ${\scriptsize \begin{bmatrix}
		I_n \\
		\theta
		\end{bmatrix}}$. However, this ambiguity does not affect our interests and results since the two different associations $\Gamma \mapsto M_i$ and $\Gamma \mapsto M_i'$ give maximal abelian subgroups $S_\theta$ and $S_\theta'$ which are in the same $\cal C_n^{\text{loc}}$-orbit: indeed,
	\[ J{\footnotesize \begin{bmatrix} I_n \\ \theta \end{bmatrix}} \ = \ {\footnotesize \begin{bmatrix} 0 & I_n \\ I_n & 0 \end{bmatrix}\begin{bmatrix} I_n \\ \theta \end{bmatrix}} \ = \ {\footnotesize \begin{bmatrix} \theta \\ I_n \end{bmatrix}}\]
	and $J \in \Sp_{2n}^{\text{loc}}(\mathbb F_2)$. On one hand, the choice of using the $M_i'$ allows to have an immediate transcription in the variety of symmetric principal minors $\cal Z_n$: via the Lagrangian mapping, the maximal abelian subgroup $S_\theta'$ corresponds to the point $[1: \theta_{ii}: \theta_{[i,j]}: \ldots : \det \theta]\in \cal Z_n$. On the other hand, the choice of using the $M_i$ is more common in the literature. Thus in this section we are going to work with the association $\Gamma \mapsto S_\theta$, but in Section 7 we will use the one $\Gamma \mapsto S_\theta'$ for exhibiting examples.
	\end{obs}

\indent Now we wonder if the maximal abelian subgroup $S_\theta \in \scr S(\cal P_n)$ defined by a graph $\Gamma_\theta$ (with adjacent matrix $\theta$) is a maximal stabilizer group. In general, the answer is negative unless we add the assumption that the graph is {\em loopless} (a.k.a {\em simple}), i.e. $\theta_{ii}=0$ for any $i=1:n$. This fact seems to be well known and implicit in the literature but we haven't be able to find a proof, thus we propose one.

\begin{prop}
	Let $\Gamma_\theta$ be a (non-oriented) graph and let $S_\theta \in \scr S(\cal P_n)$ be the associated maximal abelian group $S_\theta$ via \eqref{graph observables}. Then $S_\theta$ is a maximal stabilizer group {\em if and only if} the graph $\Gamma_\theta$ is loopless:
	\[ S_\theta \in \scr S_n^+ \ \iff \ \theta_{ii}=0 \ \ , \forall i=1:n \ . \]
\end{prop}
\begin{proof}
	Recall that $S_\theta \in \scr S_n^+ \ \iff \ \left( \ -I^{\otimes n}\notin S_\theta \ \ \wedge \ \ S_\theta \in \scr S(\cal P_n) \ \right)$.\\
	\indent $[\Rightarrow]$ Assume that $S_\theta \in \scr S_n^+$ and that, by contradiction, the graph has at least one loop, that is there exists $i_0 \in \{1,\ldots, n\}$ such that $\theta_{i_0i_0}=1$. Then the element $M_{i_0}=Z^{\theta_{i_01}}\otimes \ldots \otimes ZX \otimes \ldots \otimes Z^{\theta_{i_0n}}= Z^{\theta_{i_01}}\otimes \ldots \otimes (iY) \otimes \ldots \otimes Z^{\theta_{i_0n}}$ squares to $M_{i_0}^2=I\otimes \ldots \otimes (iY)^2 \otimes \ldots \otimes I= -I^{\otimes n}$, hence $-I^{\otimes n}\in S_\theta$, leading to a contradiction.$_\rightsquigarrow$\\
	\indent $[\Leftarrow]$ Assume $\Gamma_\theta$ to be loopless. Then for any $i=1:n$ it holds $M_i^2=I^{\otimes n}$, that is $M_i$ can only have eigenvalues $+1$ and $-1$. Since the $M_i$'s are mutually commuting, there exists (at least) one common eigenvector, say $|\gamma\rangle\in (\mathbb C^2)^{\otimes n}$: then for any $i=1:n$ we get $M_i|\gamma\rangle = (-1)^{c_i}|\gamma\rangle$. In particular, for any $i=1:n$ it holds $(-1)^{c_i}M_i|\gamma\rangle = |\gamma\rangle$, that is $|\gamma\rangle$ is common eigenvector with eigenvalue $+1$ of the $n$ elements $(-1)^{c_i}M_i$'s. But we can always find a local Clifford element $U\in \cal C_n^{\text{loc}}$ such that $U(-1)^{c_i}M_iU^\dagger=M_i$ for all $i=1:n$: if $c_{i_1}, \ldots , c_{i_k}$ are the only non-zero exponents, then by \eqref{C_1 on P_n} the local Clifford tranformation $\hat U=H_{i_1}\cdot \ldots \cdot H_{i_k}$ (where $H$ is the Hadamard matrix and $H_j=I \otimes \ldots \otimes \underbrace{H}_{j-th}\otimes \ldots \otimes I$) is the one doing the work. It follows that $\hat U|\gamma\rangle $ is common eigenvector with eigenvalue $+1$ of $M_1,\ldots, M_n$, hence of $S_{\theta}\in \scr S_n^+$.
\end{proof}

Since we are interested in the so-called {\em graph states} (which are stabilizer states) and how $\cal C_n^{\text{loc}}\rtimes \mathfrak S_n$ acts on them, {\em from now on we restrict to considering only the maximal abelian subgroups which are defined by loopless graphs}: we underline once again that graphs with loops can only define maximal abelian subgroups which are not stabilizer, thus in these cases one cannot talk about graph states as stabilizer states.

\begin{Def}
	A ($n$-qubit) {\bf graph state} $|\Gamma\rangle\in (\mathbb C^2)^{\otimes n}$ is the stabilizer state (i.e. common eigenvector with eigenvalue $+1$) of a maximal stabilizer group $S_\theta\in \scr S_n^+$ defined by a $n$-vertex ({\em loopless}) graph $\Gamma_\theta$. We denote the subset of $n$-qubit graph states by $\Theta_n\subset \Phi_1^n$.\\
	Moreover, we define a {\bf graph group} to be a maximal stabilizer group defined by a ({\em loopless}) graph and we denote the set of such subgroups by $\scr S_{\Theta_n} \subset \scr S_n^+$.
	\end{Def}

The one-to-one correspondence $\scr S_n^+ \leftrightarrow \Phi_1^n$ restricts to a one-to-one correspondence $\scr S_{\Theta_n} \leftrightarrow \Theta_n$. Next, we wonder if the actions of $\mathfrak S_n$ and $\cal C_n^{\text{loc}}$ on $\scr S_n^+$ preserve $\Theta_n$: let us investigate the two actions separately. 

\paragraph{Action of $\mathfrak S_n$.} Let $|\Gamma\rangle \in \Theta_n$ be a graph state defined by a graph $\Gamma$ with adjacent matrix $\theta$. The graph group $S_\theta=\langle M_1,\ldots , M_n\rangle \in \scr S_{\Theta_n}$ (as in \eqref{graph observables}) is described by the $2n\times n$ matrix {\scriptsize $\begin{bmatrix} \theta \\I_n\end{bmatrix}$}. Given $\sigma \in \mathfrak S_n$, by \eqref{Sn on matrices} we know that $\sigma\cdot {\scriptsize \begin{bmatrix} \theta \\I_n\end{bmatrix}}={\scriptsize \begin{bmatrix} ^\sigma\!\theta \\I_n\end{bmatrix}}$ which describes the subgroup $\sigma \cdot S_\theta= \ ^\sigma\!(S_\theta)=\langle ^\sigma\!M_{\sigma(1)}, \ldots , \ ^\sigma\!M_{\sigma(n)}\rangle$: the matrix $^\sigma\!\theta$ is still symmetric and uniquely defines a graph $^\sigma\!\Gamma=\Gamma_{^\sigma\!\theta}$. It follows that $\mathfrak S_n$ preserves $\scr S_{\Theta_n}$ and the action \eqref{Sn on true stab group} restricts to the action
\begin{equation}\label{Sn on graph group}
\begin{matrix}
\mathfrak S_n \ \times \ \scr S_{\Theta_n} & \longrightarrow & \scr S_{\Theta_n}\\
(\sigma \ , \ S_\theta) & \mapsto & ^\sigma\!(S_\theta)=S_{^\sigma\!\theta}
\end{matrix} \ .
\end{equation}

\begin{obs}\label{graphical Sn action}
	The action \eqref{Sn on graph group} reflects an action on the $n$-graphs described as follows: given a graph $\Gamma$ and a permutation $\sigma \in \mathfrak S_n$, the graph $^\sigma\!\Gamma$ is obtained simply by renaming the vertices from $\{1,\ldots,n\}$ to $\{\sigma(1),\ldots, \sigma(n)\}$, without changing the edges. Roughly speaking, two graphs are in the same $\mathfrak S_n$-orbit if they have the same drawing representation as \textquotedblleft dots-edges\textquotedblright .
	\end{obs}

\indent By correspondence, it follows that the $\mathfrak S_n$-action \eqref{Sn on true stab state} preserves the subset of graph states $\Theta_n$, hence
\begin{equation}\label{Sn on graph state}
\begin{matrix}
\frak S_n \ \times \ \Theta_n & \longrightarrow & \Theta_n\\
(\sigma \ , \ |\Gamma\rangle) & \mapsto & ^\sigma\!|\Gamma\rangle=|^\sigma\!\Gamma\rangle
\end{matrix}
\end{equation}
where $|^\sigma\Gamma\rangle$ is obtained simply by permuting via $\sigma$ the tensor entries of $|\Gamma\rangle$ as in \eqref{Sn on stab state}.

\paragraph{Action of $\cal C_n^{\text{loc}}$.} As already spoilered in Remark \ref{rmk improper orbits}, the situation with respect to the action of the local Clifford group is a little more tricky. Given a local Clifford transformation $U \in \cal C_n^{\text{loc}}$ and a graph group $S_\theta\in \scr S_{\Theta_n}$, we can consider the corresponding local symplectic transformation $\tilde U={\tiny \begin{bmatrix} A & B \\ C & D \end{bmatrix}}\in \Sp_{2n}^{\text{loc}}(\mathbb F_2)$ of type \eqref{locsymplform} and the matrix {\scriptsize $\begin{bmatrix} \theta \\ I_n\end{bmatrix}$} describing $S_\theta$: then the action $U \cdot S_\theta$ corresponds to the linear transformation 
\[ \tilde U \cdot \begin{bmatrix} \theta \\ I_n \end{bmatrix}= \begin{bmatrix}
	A\theta+ B\\
	C\theta+ D
	\end{bmatrix}\]
that is, in general, $U\cdot S_\theta$ is not a graph group (but a maximal stabilizer group of course). This explains Remark \ref{rmk improper orbits}.\\ 
However, the local Clifford operations preserving $\Theta_n$ are known \cite[Proposition 5.3]{hein2006} to correspond to operations on $n$-graphs, called {\em local complementations}: for a formal definition and a graphical description we refer to \cite[Sec. 1]{bouchet1993}, \cite{hein2006} and \cite[Sec. IV, VI]{article1}.

\begin{obs}
	The local complementations together with Remark \ref{graphical Sn action} give a graphical description of the action of (a subgroup of) $\cal C_n^{\text{loc}}\rtimes \frak S_n$ on the graph states.
	\end{obs}

If we denote by $G_n <\cal C_n^{\text{loc}}$ the subgroup of the local complementations, and by $u$ the graph transformation corresponding to $U \in G_n$, then one gets the actions

\begin{equation}\label{G_n on graph}
\begin{matrix}
G_n \ \times \ \scr S_{\Theta_n} & \longrightarrow & \scr S_{\Theta_n}\\
(U \ , \ S_\theta) & \mapsto & S_{u\cdot \theta}
\end{matrix}
\ \ \ \ \ , \ \ \ \ \
\begin{matrix}
G_n \ \times \ \Theta_n & \longrightarrow & \Theta_n\\
(U \ , \ |\Gamma\rangle) & \mapsto & U|\Gamma\rangle
\end{matrix}
\end{equation}

\noindent which, together with \eqref{Sn on graph group} and \eqref{Sn on graph state}, give the actions of the semidirect product $G_n\rtimes \mathfrak S_n$
\begin{equation}\label{semidir on graph}
\begin{matrix}
\big( G_n\rtimes \mathfrak S_n\big) \ \times \scr S_{\Theta_n} & \longrightarrow & \scr S_{\Theta_n}\\
\big((U,\sigma) \ , \ S_{\theta}\big) & \mapsto & S_{u\cdot (^\sigma\!\theta)}
\end{matrix} \ \ \ \ , \ \ \ \ 
\begin{matrix}
\big(G_n\rtimes \mathfrak S_n\big) \times \ \Theta_n & \longrightarrow & \Theta_n\\
\big( (U,\sigma) \ ,\ |\Gamma \big) & \mapsto & U|^\sigma\!\Gamma\rangle
\end{matrix} \ .
\end{equation}

\indent We have seen that, for any graph group $S_\theta\in \scr S_{\Theta_n}$ and any $U\in \cal C_n^{\text{loc}}$, the subgroup $U\cdot S_\theta$ is still a stabilizer group, thus {\em any graph group (resp. graph state) is the representative of a $\cal C_n^{\text{loc}}$-orbit of maximal stabilizer groups (resp. stabilizer states)}. But we can say more: by \cite[Theorem 1]{article1}, each maximal stabilizer group (resp. stabilizer state) is $\cal C_n^{\text{loc}}$-equivalent to a graph group (resp. graph state), thus {\em any $\cal C_n^{\text{loc}}$-orbit of stabilizer states admits a representative which is a graph states}. This means that studying the orbits of stabilizer states under the action of $\cal C_n^{\text{loc}}\rtimes \mathfrak S_n$ is the same as studying the orbits of graph states under the action of $G_n\rtimes \mathfrak S_n$:

\begin{equation}\label{orbits bij of stab and graph}
\faktor{\Phi_1^n}{ C_n^{\text{loc}}\rtimes \mathfrak S_n} \  \longleftrightarrow \ \faktor{\Theta_n}{G_n \rtimes \mathfrak S_n} \ .
\end{equation}

\paragraph{Conclusion:} This section achieves the proof of Theorem \ref{theo:main}, proving that the classification (up to $\cal C_n^{\text{loc}}\rtimes \mathfrak S_n$ action) of (loopless) $n$-graph states is in a one-to-one correspondence with the orbits in $\cal Z_n$ under the action of $\SL(2,\mathbb F_2)^{\times n}\rtimes \mathfrak S_n$: by putting together the correspondences \eqref{orbits bij of stab and Zn} and \eqref{orbits bij of stab and graph} we get
\begin{equation}\label{orbits bij of graph and Zn}
\faktor{\Theta_n}{ G_n\rtimes_\phi \mathfrak S_n} \   \longleftrightarrow \ \faktor{\cal Z_n}{\SL(2,\mathbb F_2)^{\times n}\rtimes \mathfrak S_n} \ .
\end{equation}

\section{Applications}\label{sec:application}

\indent \indent We propose two applications of Theorem \ref{theo:main}. First, we show how the $4$-qubit graph states classification can be deduced from the orbit stratification of $\mathcal{Z}_4$. Then, in the other direction, we show how the knowledge of the $5$-qubit graph states classification can be used to obtain representatives of the orbits in $\mathcal{Z}_5$.

\subsection{Classification of $4$-qubit graph states from orbits of $\cal Z_4$}

\indent \indent Under the action of $\SL(2,\mathbb F_2)^{\times 4}\rtimes \frak S_4$, the variety of principal minor $\cal Z_4 \subset \mathbb P_2^{15}$ splits in six orbits whose cardinalities and representatives are listed in \cite[Table 3]{holweck2014}. Starting from that classification, we recover the orbits (in the sense of Remark \ref{rmk improper orbits}) of $4$-graphs under the action of $\cal C_4^{\text{loc}}\rtimes \frak S_4$. 

\begin{obs}
	In order to be faithful to the notations we used in the first part of our work, we will consider representatives of the orbits $\cal O_2, \cal O_3, \cal O_6 , \cal O_{14}, \cal O_{17}, \cal O_{18}$ different from the ones in Table 3 \cite{holweck2014}. More precisely, we will take representatives in the chart $\{z_0\neq 0\}$ in order to recover more easily the corresponding graph-matrices. Moreover, since we want to obtain loopless graphs, we choose representatives whose coordinates $z_1=z_2=z_3=z_4=0$ are zero: we denote by $\bold{\underline{0}}$ those entries.
	\end{obs}

\begin{equation}\label{orbits 4-graphs}
{\footnotesize \begin{matrix}
\hline
\text{Orbit} & \text{Representative in $\cal Z_4\subset \mathbb P_2^7$} & \text{Graph-matrix} & \text{$4$-Graph}\\

\hline
\cal O_2 
&
[1:\bold{\underline{0}}:0:0:0:0:0:0:0:0:0 : 0:0]
&
{\tiny \begin{bmatrix}
	0 \\
	  & 0\\
	  &   & 0 \\
	  &   &   & 0 
	\end{bmatrix}}
&
\begin{minipage}{0.1\textwidth}
\begin{tikzpicture}[scale=0.75]
\node(1) at (-0.7,0.7){$\bullet$};
\node(2) at (0.7,0.7){$\bullet$};
\node(3) at (0.7,-0.7){$\bullet$};
\node(4) at (-0.7,-0.7){$\bullet$};

\path[font=\scriptsize, >= angle 90]
;
\end{tikzpicture}
\end{minipage}
\\

\hline
\cal O_3
&
[1:\bold{\underline{0}}:1:0:0:0:0:0:0:0:0 : 0:0]
&
{\tiny \begin{bmatrix}
	0 & 1 \\
	1 & 0\\
	&   & 0 \\
	&   &   & 0 
	\end{bmatrix}}
&
\begin{minipage}{0.1\textwidth}
\begin{tikzpicture}[scale=0.75]
\node(1) at (-0.7,0.7){$\bullet$};
\node(2) at (0.7,0.7){$\bullet$};
\node(3) at (0.7,-0.7){$\bullet$};
\node(4) at (-0.7,-0.7){$\bullet$};

\path[font=\scriptsize, >= angle 90]
(1) edge [] node [] {} (2);
\end{tikzpicture}
\end{minipage}
\\

\hline
\cal O_6
&
[1: \bold{\underline{0}}: 1: 0:0:1: 0:0:0:0:0 : 0:0]
&
{\tiny \begin{bmatrix}
	0 & 1 \\
	1 & 0 & 1\\
	  & 1 & 0 \\
	&   &   & 0 
	\end{bmatrix}}
&
\begin{minipage}{0.1\textwidth}
\begin{tikzpicture}[scale=0.75]
\node(1) at (-0.7,0.7){$\bullet$};
\node(2) at (0.7,0.7){$\bullet$};
\node(3) at (0.7,-0.7){$\bullet$};
\node(4) at (-0.7,-0.7){$\bullet$};

\path[font=\scriptsize, >= angle 90]
(1) edge [] node [] {} (2)
(2) edge [] node [] {} (3);
\end{tikzpicture}
\end{minipage}
\\

\hline
\cal O_{14}
&
[1:\bold{\underline{0}}:1:1:1:0:0:0:0:0:0 : 0:0]
&
{\tiny \begin{bmatrix}
	0 & 1 & 1 & 1 \\
	1 & 0 &  \\
	1 &   & 0 & \\
	1 &   &   & 0 
	\end{bmatrix}}
&
\begin{minipage}{0.1\textwidth}
\begin{tikzpicture}[scale=0.75]
\node(1) at (-0.7,0.7){$\bullet$};
\node(2) at (0.7,0.7){$\bullet$};
\node(3) at (0.7,-0.7){$\bullet$};
\node(4) at (-0.7,-0.7){$\bullet$};

\path[font=\scriptsize, >= angle 90]
(1) edge [] node [] {} (2)
(1) edge [] node [] {} (3)
(1) edge [] node [] {} (4);
\end{tikzpicture}
\end{minipage}
\\

\hline
\cal O_{17}
&
[1: \bold{\underline{0}}: 1: 0:0:0:0:1: 0: 0:0:0:1]
&
{\tiny \begin{bmatrix}
	0 & 1\\
	1 & 0 &  \\
	&   & 0 & 1\\
	&   & 1 & 0 
	\end{bmatrix}}
&
\begin{minipage}{0.1\textwidth}
\begin{tikzpicture}[scale=0.75]
\node(1) at (-0.7,0.7){$\bullet$};
\node(2) at (0.7,0.7){$\bullet$};
\node(3) at (0.7,-0.7){$\bullet$};
\node(4) at (-0.7,-0.7){$\bullet$};

\path[font=\scriptsize, >= angle 90]
(1) edge [] node [] {} (2)
(3) edge [] node [] {} (4);
\end{tikzpicture}
\end{minipage}
\\

\hline
\cal O_{18}
&
[1:\bold{\underline{0}}: 1: 0: 1: 1: 0: 1: 0: 0: 0:0:0]
&
{\tiny \begin{bmatrix}
	0 & 1 &   & 1\\
	1 & 0 & 1\\
	  & 1 & 0 & 1\\
	1 &   & 1 & 0 
	\end{bmatrix}}
&
\begin{minipage}{0.1\textwidth}
\begin{tikzpicture}[scale=0.75]
\node(1) at (-0.7,0.7){$\bullet$};
\node(2) at (0.7,0.7){$\bullet$};
\node(3) at (0.7,-0.7){$\bullet$};
\node(4) at (-0.7,-0.7){$\bullet$};

\path[font=\scriptsize, >= angle 90]
(1) edge [] node [] {} (2)
(2) edge [] node [] {} (3)
(3) edge [] node [] {} (4)
(4) edge [] node [] {} (1);
\end{tikzpicture}
\end{minipage}
\end{matrix}}
\end{equation}

\begin{obs}
	The vertices of the representative graphs in Table \eqref{orbits 4-graphs} are not numbered because of Remark \ref{graphical Sn action}: as representative labeling, we choose the vertex $1$ to be the upper-left one and the other vertices are clockwise ordered.
	\end{obs}

By the Lagrangian mapping (Sec. \ref{preliminaries}), we know that for each orbit $\mathcal{O}_i$ of $\mathcal{Z}_4$ there are $4^4|\mathcal{O}_i|$ corresponding stabilizer states. Indeed for each isotropic $4$-dimensional space in $\mathbb F_2^{2n}$, there are $4$ different choices of the phases for the four elements that span it and therefore $4^4$ different stablizer states. Thus, by the sizes in Table 3 \cite{holweck2014} we recover the number of stabilizer states coming from each graph orbit.

\begin{equation}
{\small \begin{matrix}
\hline
\text{Orbit} & \text{Size of orbit in $\cal Z_4$} & \text{\# of stabilizer states}\\
\hline
\cal O_2 & 81 & 20736  \\
\hline
\cal O_3 & 324 & 82944  \\
\hline
\cal O_6 & 648 & 165888  \\
\hline
\cal O_{14} & 162 & 41472  \\
\hline
\cal O_{17} & 108 &  27648 \\
\hline 
\cal O_{18} & 972 &   248832
\end{matrix}}
\end{equation}

\subsection{Orbits of $\cal Z_5$ from $5$-qubit graph states classification}

\indent \indent In \cite{danielsen2005} the number of orbits (in the sense of Remark \ref{rmk improper orbits}) of $5$-graphs was computed to be $11$ (see also \cite[Table IV]{hein2006}). In the tables \eqref{5-orbits of stab-groups} and \eqref{5-orbits of Z5}, we exhibit $11$ representatives of such orbits and by them we recover the representatives of the orbits in $\cal Z_5 \subset \mathbb P(\mathbb F_2^{32})$.

\begin{obs}
	The representative graphs in Table \eqref{5-orbits of stab-groups} and Table \eqref{5-orbits of Z5} are loopless and their vertices are not numbered because of Remark \ref{graphical Sn action}: as representative labeling, we choose the vertex $1$ to be the highest one and the other vertices are clockwise ordered. \\
	\indent In the last column of Table \eqref{5-orbits of stab-groups}, the observables generating the stabilizer groups are obtained by the columns of the matrix ${\scriptsize \begin{bmatrix} I_5\\\theta\end{bmatrix}}$ where $\theta$ is the matrix in the third column.\\
	\indent In the last column of Table \eqref{5-orbits of Z5}, the coordinates of the points in $\cal Z_5\subset \mathbb P_2^{31}$ are given by the principal minors of the matrix $\theta$ (third column of Table \eqref{5-orbits of stab-groups}). In particular, the points are taken in the chart $\{z_0\neq 0\}$ and the entries $z_{1}=\ldots =z_{5}=0$ are all zero since the representative graphs are loopless (i.e. $\theta_{ii}=0$ for all $i=1:5$): for layout reasons, we compactly write such five coordinates as a zero vector $\bold{\underline{0}}$.
\end{obs}

\begin{obs}
	Another representative of the orbit $\cal O_9$ is the $5$-graph of the GHZ state
	
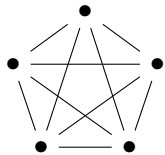
\[   
	\begin{tikzpicture}[scale=1]
	\node(1) at (0,1){$\bullet$};
	\node(2) at (0.95,0.3){$\bullet$};
	\node(3) at (0.58,-0.8){$\bullet$};
	\node(4) at (-0.58,-0.8){$\bullet$};
	\node(5) at (-0.95,0.3){$\bullet$};
	
	\path[font=\scriptsize, >= angle 90]
	(1) edge [] node [] {} (2)
	(1) edge [] node [] {} (3)
	(1) edge [] node [] {} (4)
	(1) edge [] node [] {} (5)
	(2) edge [] node [] {} (3)
	(3) edge [] node [] {} (4)
	(4) edge [] node [] {} (5)
	(2) edge [] node [] {} (4)
	(2) edge [] node [] {} (5)
	(3) edge [] node [] {} (5);
	\end{tikzpicture}
	\]
\end{obs}

\section*{Acknowledgement}
\addcontentsline{toc}{section}{Acknowledgement}
The work presented in this paper started when F.H. was a guest Professor at Trento University in October 2020 within the Quantum@Trento program. F.H. would like to warmly thank Alessandra Bernardi and Iacopo Carusotto for making this invitation possible. This  work  was partially  supported  by  the  French  Investissements  d’Avenir  programme,  project ISITE-BFC (contract ANR-15-IDEX-03). V.G. is partially supported by Italian GNSAGA-INDAM.

\begin{equation}\label{5-orbits of stab-groups}
{\small \begin{matrix}
\hline
\text{Orbit} &
\text{$5$-Graph} & \text{Graph-matrix} & \text{Stabilizer-group in $\cal P_5$} 
\\
\hline
\cal O_1
&
\begin{minipage}{0.13\textwidth}
\begin{tikzpicture}[scale=0.75]
\node(1) at (0,1){$\bullet$};
\node(2) at (0.95,0.3){$\bullet$};
\node(3) at (0.58,-0.8){$\bullet$};
\node(4) at (-0.58,-0.8){$\bullet$};
\node(5) at (-0.95,0.3){$\bullet$};

\path[font=\scriptsize, >= angle 90]
;
\end{tikzpicture}
\end{minipage}
&
{\tiny \begin{bmatrix}
0\\
& 0 \\
& & 0\\
& & & 0 \\
& & & & 0
\end{bmatrix}}
&
\langle ZIIII, IZIII, IIZII, IIIZI, IIIIZ \rangle

\\
\hline
\cal O_2
&
\begin{minipage}{0.13\textwidth}
\begin{tikzpicture}[scale=0.75]
\node(1) at (0,1){$\bullet$};
\node(2) at (0.95,0.3){$\bullet$};
\node(3) at (0.58,-0.8){$\bullet$};
\node(4) at (-0.58,-0.8){$\bullet$};
\node(5) at (-0.95,0.3){$\bullet$};

\path[font=\scriptsize, >= angle 90]
(1) edge [] node [] {} (2);
\end{tikzpicture}
\end{minipage}
&
{\tiny \begin{bmatrix}
	0 & 1\\
	1 & 0 \\
	& & 0\\
	& & & 0 \\
	& & & & 0
	\end{bmatrix}}
&
\langle ZXIII, XZIII, IIZII, IIIZI, IIIIZ \rangle 

\\
\hline
\cal O_3
&
\begin{minipage}{0.13\textwidth}
\begin{tikzpicture}[scale=0.75]
\node(1) at (0,1){$\bullet$};
\node(2) at (0.95,0.3){$\bullet$};
\node(3) at (0.58,-0.8){$\bullet$};
\node(4) at (-0.58,-0.8){$\bullet$};
\node(5) at (-0.95,0.3){$\bullet$};

\path[font=\scriptsize, >= angle 90]
(1) edge [] node [] {} (2)
(2) edge [] node [] {} (3);
\end{tikzpicture}\end{minipage}
&
{\tiny \begin{bmatrix}
	0 & 1\\
	1 & 0 & 1 \\
	  & 1 & 0\\
	  &   &   & 0 \\
  	  &   &   &   & 0
	\end{bmatrix}}
&
\langle ZXIII, XZXII, IXZII, IIIZI, IIIIZ \rangle 

\\
\hline
\cal O_4
&
\begin{minipage}{0.13\textwidth}
\begin{tikzpicture}[scale=0.75]
\node(1) at (0,1){$\bullet$};
\node(2) at (0.95,0.3){$\bullet$};
\node(3) at (0.58,-0.8){$\bullet$};
\node(4) at (-0.58,-0.8){$\bullet$};
\node(5) at (-0.95,0.3){$\bullet$};

\path[font=\scriptsize, >= angle 90]
(1) edge [] node [] {} (2)
(3) edge [] node [] {} (4);
\end{tikzpicture}\end{minipage}
&
{\tiny \begin{bmatrix}
	0 & 1\\
	1 & 0 &  \\
	  &   & 0 & 1\\
	  &   & 1 & 0 \\
	  &   &   &   & 0
	\end{bmatrix}}
&
\langle ZXIII , XZIII, IIZX, IIXZI, IIIIZ \rangle 

\\
\hline
\cal O_5
&
\begin{minipage}{0.13\textwidth}
\begin{tikzpicture}[scale=0.75]
\node(1) at (0,1){$\bullet$};
\node(2) at (0.95,0.3){$\bullet$};
\node(3) at (0.58,-0.8){$\bullet$};
\node(4) at (-0.58,-0.8){$\bullet$};
\node(5) at (-0.95,0.3){$\bullet$};

\path[font=\scriptsize, >= angle 90]
(1) edge [] node [] {} (2)
(2) edge [] node [] {} (3)
(3) edge [] node [] {} (4);
\end{tikzpicture}\end{minipage}
&
{\tiny \begin{bmatrix}
	0 & 1\\
	1 & 0 & 1 \\
	  & 1 & 0 & 1\\
	  &   & 1 & 0 \\
	  &   &   &   & 0
	\end{bmatrix}}
&
\langle ZXIII, XZXII, IXZXI, IIXZI, IIIIZ\rangle

\\
\hline
\cal O_6
&
\begin{minipage}{0.13\textwidth}
\begin{tikzpicture}[scale=0.75]
\node(1) at (0,1){$\bullet$};
\node(2) at (0.95,0.3){$\bullet$};
\node(3) at (0.58,-0.8){$\bullet$};
\node(4) at (-0.58,-0.8){$\bullet$};
\node(5) at (-0.95,0.3){$\bullet$};

\path[font=\scriptsize, >= angle 90]
(1) edge [] node [] {} (2)
(5) edge [] node [] {} (1)
(3) edge [] node [] {} (4);
\end{tikzpicture}\end{minipage}
&
{\tiny \begin{bmatrix}
	0 & 1 &   &   & 1\\
	1 & 0 &  \\
	  &   & 0 & 1\\
	  &   & 1 & 0 \\
	1 &   &   &   & 0
	\end{bmatrix}}
&
\langle ZXIIX, XZIII, IIZXI, IIXZI, XIIIZ \rangle

\\
\hline
\cal O_7
&
\begin{minipage}{0.13\textwidth}
\begin{tikzpicture}[scale=0.75]
\node(1) at (0,1){$\bullet$};
\node(2) at (0.95,0.3){$\bullet$};
\node(3) at (0.58,-0.8){$\bullet$};
\node(4) at (-0.58,-0.8){$\bullet$};
\node(5) at (-0.95,0.3){$\bullet$};

\path[font=\scriptsize, >= angle 90]
(1) edge [] node [] {} (2)
(2) edge [] node [] {} (3)
(3) edge [] node [] {} (4)
(4) edge [] node [] {} (5);
\end{tikzpicture}\end{minipage}
&
{\tiny \begin{bmatrix}
	0 & 1\\
	1 & 0 & 1 \\
	  & 1 & 0 & 1\\
	  &   & 1 & 0 & 1 \\
	  &   &   & 1 & 0
	\end{bmatrix}}
&
\langle ZXIII, XZXII, IXZXI, IIXZX, IIIXZ \rangle

\\
\hline
\cal O_8
&
\begin{minipage}{0.13\textwidth}
\begin{tikzpicture}[scale=0.75]
\node(1) at (0,1){$\bullet$};
\node(2) at (0.95,0.3){$\bullet$};
\node(3) at (0.58,-0.8){$\bullet$};
\node(4) at (-0.58,-0.8){$\bullet$};
\node(5) at (-0.95,0.3){$\bullet$};

\path[font=\scriptsize, >= angle 90]
(1) edge [] node [] {} (2)
(2) edge [] node [] {} (3)
(3) edge [] node [] {} (4)
(4) edge [] node [] {} (1);
\end{tikzpicture}\end{minipage}
&
{\tiny \begin{bmatrix}
	0 & 1 &   & 1\\
	1 & 0 & 1 \\
  	  & 1 & 0 & 1\\
	1 &   & 1 & 0 \\
	  &   &   &   & 0
	\end{bmatrix}}
&
\langle ZXIXI, XZXII, IXZXI, XIXZI, IIIIZ \rangle 

\\
\hline
\cal O_9
&
\begin{minipage}{0.13\textwidth}
\begin{tikzpicture}[scale=0.75]
\node(1) at (0,1){$\bullet$};
\node(2) at (0.95,0.3){$\bullet$};
\node(3) at (0.58,-0.8){$\bullet$};
\node(4) at (-0.58,-0.8){$\bullet$};
\node(5) at (-0.95,0.3){$\bullet$};

\path[font=\scriptsize, >= angle 90]
(1) edge [] node [] {} (2)
(1) edge [] node [] {} (3)
(1) edge [] node [] {} (4)
(1) edge [] node [] {} (5);
\end{tikzpicture}\end{minipage}
&
{\tiny \begin{bmatrix}
	0 & 1 & 1 & 1 & 1\\
	1 & 0 &  \\
	1 &   & 0 & \\
	1 &   &   & 0 \\
	1 &   &   &   & 0
	\end{bmatrix}}
&
\langle ZXXXX, XZIII, XIZII, XIIZI, XIIIZ \rangle

\\
\hline
\cal O_{10}
&
\begin{minipage}{0.13\textwidth}
\begin{tikzpicture}[scale=0.75]
\node(1) at (0,1){$\bullet$};
\node(2) at (0.95,0.3){$\bullet$};
\node(3) at (0.58,-0.8){$\bullet$};
\node(4) at (-0.58,-0.8){$\bullet$};
\node(5) at (-0.95,0.3){$\bullet$};

\path[font=\scriptsize, >= angle 90]
(1) edge [] node [] {} (2)
(2) edge [] node [] {} (3)
(3) edge [] node [] {} (4)
(4) edge [] node [] {} (5)
(5) edge [] node [] {} (1)
;
\end{tikzpicture}\end{minipage}
&
{\tiny \begin{bmatrix}
	0 & 1 &   &   & 1\\
	1 & 0 & 1 \\
	  & 1 & 0 & 1\\
 	  &   & 1 & 0 & 1\\
	1 &   &   & 1 & 0
	\end{bmatrix}}
&
\langle ZXIIX, XZXII, IXZXI, IIXZX, XIIXZ \rangle

\\
\hline
\cal O_{11}
&
\begin{minipage}{0.13\textwidth}
\begin{tikzpicture}[scale=0.75]
\node(1) at (0,1){$\bullet$};
\node(2) at (0.95,0.3){$\bullet$};
\node(3) at (0.58,-0.8){$\bullet$};
\node(4) at (-0.58,-0.8){$\bullet$};
\node(5) at (-0.95,0.3){$\bullet$};

\path[font=\scriptsize, >= angle 90]
(1) edge [] node [] {} (2)
(1) edge [] node [] {} (3)
(1) edge [] node [] {} (4);
\end{tikzpicture}\end{minipage}
&
{\tiny \begin{bmatrix}
	0 & 1 & 1 & 1 & \\
	1 & 0 &  \\
	1 &   & 0 & \\
	1 &   &   & 0 \\
	 &   &   &   & 0
	\end{bmatrix}}
&
\langle ZXXXI, XZIII, XIZII, XIIZI, IIIIZ \rangle
\end{matrix}}
\end{equation}

\begin{equation}\label{5-orbits of Z5}
{\footnotesize \begin{matrix}
\hline
\text{Orbit} &
\text{$5$-Graph} & \text{Representative point in $\cal Z_5\subset \mathbb P_2^{31}$}
\\
\hline
\cal O_1
&
\begin{minipage}{0.13\textwidth}
\begin{tikzpicture}[scale=0.75]
\node(1) at (0,1){$\bullet$};
\node(2) at (0.95,0.3){$\bullet$};
\node(3) at (0.58,-0.8){$\bullet$};
\node(4) at (-0.58,-0.8){$\bullet$};
\node(5) at (-0.95,0.3){$\bullet$};

\path[font=\scriptsize, >= angle 90]
;
\end{tikzpicture}
\end{minipage}

&
[1:0: \ldots : 0] 

\\
\hline
\cal O_2
&
\begin{minipage}{0.13\textwidth}
\begin{tikzpicture}[scale=0.75]
\node(1) at (0,1){$\bullet$};
\node(2) at (0.95,0.3){$\bullet$};
\node(3) at (0.58,-0.8){$\bullet$};
\node(4) at (-0.58,-0.8){$\bullet$};
\node(5) at (-0.95,0.3){$\bullet$};

\path[font=\scriptsize, >= angle 90]
(1) edge [] node [] {} (2);
\end{tikzpicture}
\end{minipage}

&
[1: \bold{\underline{0}} : \underbrace{1}_{z_6}: 0 : \ldots : 0]

\\
\hline
\cal O_3
&
\begin{minipage}{0.13\textwidth}
\begin{tikzpicture}[scale=0.75]
\node(1) at (0,1){$\bullet$};
\node(2) at (0.95,0.3){$\bullet$};
\node(3) at (0.58,-0.8){$\bullet$};
\node(4) at (-0.58,-0.8){$\bullet$};
\node(5) at (-0.95,0.3){$\bullet$};

\path[font=\scriptsize, >= angle 90]
(1) edge [] node [] {} (2)
(2) edge [] node [] {} (3);
\end{tikzpicture}\end{minipage}

&
[1: \bold{\underline{0}} : \underbrace{1}_{z_6}: 0 :0:0: 1: 0 : \ldots : 0 ]

\\
\hline
\cal O_4
&
\begin{minipage}{0.13\textwidth}
\begin{tikzpicture}[scale=0.75]
\node(1) at (0,1){$\bullet$};
\node(2) at (0.95,0.3){$\bullet$};
\node(3) at (0.58,-0.8){$\bullet$};
\node(4) at (-0.58,-0.8){$\bullet$};
\node(5) at (-0.95,0.3){$\bullet$};

\path[font=\scriptsize, >= angle 90]
(1) edge [] node [] {} (2)
(3) edge [] node [] {} (4);
\end{tikzpicture}\end{minipage}

&
[1:  \bold{\underline{0}} : \underbrace{1}_{z_6}: 0 : \ldots : 0 : \underbrace{1}_{z_{13}}: 0 : \ldots : 0: \underbrace{1}_{z_{26}}:0:\ldots : 0 ]

\\
\hline
\cal O_5
&
\begin{minipage}{0.13\textwidth}
\begin{tikzpicture}[scale=0.75]
\node(1) at (0,1){$\bullet$};
\node(2) at (0.95,0.3){$\bullet$};
\node(3) at (0.58,-0.8){$\bullet$};
\node(4) at (-0.58,-0.8){$\bullet$};
\node(5) at (-0.95,0.3){$\bullet$};

\path[font=\scriptsize, >= angle 90]
(1) edge [] node [] {} (2)
(2) edge [] node [] {} (3)
(3) edge [] node [] {} (4);
\end{tikzpicture}\end{minipage}

&
[1: \bold{\underline{0}} : \underbrace{1}_{z_6}: 0 : 0 : 0 : 1: 0 :0:1:0 \ldots : 0:\underbrace{1}_{z_{26}}:0:\ldots :0 ] 

\\
\hline
\cal O_6
&
\begin{minipage}{0.13\textwidth}
\begin{tikzpicture}[scale=0.75]
\node(1) at (0,1){$\bullet$};
\node(2) at (0.95,0.3){$\bullet$};
\node(3) at (0.58,-0.8){$\bullet$};
\node(4) at (-0.58,-0.8){$\bullet$};
\node(5) at (-0.95,0.3){$\bullet$};

\path[font=\scriptsize, >= angle 90]
(1) edge [] node [] {} (2)
(5) edge [] node [] {} (1)
(3) edge [] node [] {} (4);
\end{tikzpicture}\end{minipage}

&
[1:  \bold{\underline{0}} : \underbrace{1}_{z_6}: 0 :0:1: 0 : 0 :0: 1: 0 : \ldots : 0:\underbrace{1}_{z_{26}}:0: \ldots : 0 ] 

\\
\hline
\cal O_7
&
\begin{minipage}{0.13\textwidth}
\begin{tikzpicture}[scale=0.75]
\node(1) at (0,1){$\bullet$};
\node(2) at (0.95,0.3){$\bullet$};
\node(3) at (0.58,-0.8){$\bullet$};
\node(4) at (-0.58,-0.8){$\bullet$};
\node(5) at (-0.95,0.3){$\bullet$};

\path[font=\scriptsize, >= angle 90]
(1) edge [] node [] {} (2)
(2) edge [] node [] {} (3)
(3) edge [] node [] {} (4)
(4) edge [] node [] {} (5);
\end{tikzpicture}\end{minipage}
&
[1: \bold{\underline{0}}: \underbrace{1}_{z_6}: 0 : 0 : 0 : 1: 0 :0:1:0:1:0: \ldots : 0: \underbrace{1}_{z_{26}}:0:1:0:1:0 ] 

\\
\hline
\cal O_8
&
\begin{minipage}{0.13\textwidth}
\begin{tikzpicture}[scale=0.75]
\node(1) at (0,1){$\bullet$};
\node(2) at (0.95,0.3){$\bullet$};
\node(3) at (0.58,-0.8){$\bullet$};
\node(4) at (-0.58,-0.8){$\bullet$};
\node(5) at (-0.95,0.3){$\bullet$};

\path[font=\scriptsize, >= angle 90]
(1) edge [] node [] {} (2)
(2) edge [] node [] {} (3)
(3) edge [] node [] {} (4)
(4) edge [] node [] {} (1);
\end{tikzpicture}\end{minipage}

&
[1: \bold{\underline{0}} : \underbrace{1}_{z_6}: 0 : 1 : 0 : 1: 0 :0:1:0: \ldots : 0 ] 

\\
\hline
\cal O_9
&
\begin{minipage}{0.13\textwidth}
\begin{tikzpicture}[scale=0.75]
\node(1) at (0,1){$\bullet$};
\node(2) at (0.95,0.3){$\bullet$};
\node(3) at (0.58,-0.8){$\bullet$};
\node(4) at (-0.58,-0.8){$\bullet$};
\node(5) at (-0.95,0.3){$\bullet$};

\path[font=\scriptsize, >= angle 90]
(1) edge [] node [] {} (2)
(1) edge [] node [] {} (3)
(1) edge [] node [] {} (4)
(1) edge [] node [] {} (5);
\end{tikzpicture}\end{minipage}

&
[1: \bold{\underline{0}} : \underbrace{1}_{z_6}: 1 :1:1: 0 : \ldots : 0 ]

\\
\hline
\cal O_{10}
&
\begin{minipage}{0.13\textwidth}
\begin{tikzpicture}[scale=0.75]
\node(1) at (0,1){$\bullet$};
\node(2) at (0.95,0.3){$\bullet$};
\node(3) at (0.58,-0.8){$\bullet$};
\node(4) at (-0.58,-0.8){$\bullet$};
\node(5) at (-0.95,0.3){$\bullet$};

\path[font=\scriptsize, >= angle 90]
(1) edge [] node [] {} (2)
(2) edge [] node [] {} (3)
(3) edge [] node [] {} (4)
(4) edge [] node [] {} (5)
(5) edge [] node [] {} (1)
;
\end{tikzpicture}\end{minipage}

&
[1: \bold{\underline{0}} : \underbrace{1}_{z_6}: 0 : 0 : 1 : 1: 0 :0:1:0:1:0: \ldots : 0:\underbrace{1}_{z_{27}}:1:1:1: 0 ]  

\\
\hline
\cal O_{11}
&
\begin{minipage}{0.13\textwidth}
\begin{tikzpicture}[scale=0.75]
\node(1) at (0,1){$\bullet$};
\node(2) at (0.95,0.3){$\bullet$};
\node(3) at (0.58,-0.8){$\bullet$};
\node(4) at (-0.58,-0.8){$\bullet$};
\node(5) at (-0.95,0.3){$\bullet$};

\path[font=\scriptsize, >= angle 90]
(1) edge [] node [] {} (2)
(1) edge [] node [] {} (3)
(1) edge [] node [] {} (4);
\end{tikzpicture}\end{minipage}

&
[1: \bold{\underline{0}} : \underbrace{1}_{z_6}: 1 :1: 0 : \ldots : 0 ]

\end{matrix}}
\end{equation}


\addcontentsline{toc}{section}{References}

\printbibliography

@inproceedings{transvection,
	title={Synthesis of logical Clifford operators via symplectic geometry},
	author={Rengaswamy, N. and Calderbank, R. and Pfister, H. D and Kadhe, S.},
	booktitle={2018 IEEE International Symposium on Information Theory (ISIT)},
	pages={791--795},
	year={2018},
	organization={IEEE}
}

@article{bouchet1993,
	title={Recognizing locally equivalent graphs},
	author={Bouchet, A.},
	journal={Discrete Mathematics},
	volume={114},
	number={1-3},
	pages={75--86},
	year={1993},
	publisher={Elsevier}
}

@article{dangniam2020optimal,
	title={Optimal verification of stabilizer states},
	author={Dangniam, N. and Han, Yun-G. and Zhu, H.},
	journal={Physical Review Research},
	volume={2},
	number={4},
	pages={043323},
	year={2020},
	publisher={APS}
}

@article{koenig,
	title={How to efficiently select an arbitrary Clifford group element},
	author={Koenig, R. and Smolin, J.A},
	journal={Journal of Mathematical Physics},
	volume={55},
	number={12},
	pages={122202},
	year={2014},
	publisher={AIP Publishing LLC}
}

@article{article1,
	title={Graphical description of the action of local Clifford transformations on graph states},
	author={Van den Nest, M. and Dehaene, J. and De Moor, B.},
	journal={Physical Review A},
	volume={69},
	number={2},
	pages={022316},
	year={2004},
	publisher={APS}
}

@article{gottesman1997,
  title={Stabilizer codes and quantum error correction},
  author={Gottesman, D.},
  journal={PhD Thesis},
  year={1997}
}

@article{hein2006,
  title={Entanglement in graph states and its applications},
  author={Hein, M. and D\"{u}r, W. and Eisert, J. and Raussendorf, R. and Nest, M. and Briegel, H-J},
  journal={arXiv preprint quant-ph/0602096},
  year={2006}
}

@article{briegel2009,
  title={Measurement-based quantum computation},
  author={Briegel, H.J and Browne, D.E and D\"{u}r, W. and Raussendorf, R. and Van den Nest, M.},
  journal={Nature Physics},
  volume={5},
  number={1},
  pages={19--26},
  year={2009},
  publisher={Nature Publishing Group}
}

@article{markham2008,
  title={Graph states for quantum secret sharing},
  author={Markham, D. and Sanders, B.C.},
  journal={Physical Review A},
  volume={78},
  number={4},
  pages={042309},
  year={2008},
  publisher={APS}
}

@article{bell2014,
  title={Experimental demonstration of graph-state quantum secret sharing},
  author={Bell, BA and Markham, D. and Herrera-Martí, DA and Marin, A. and Wadsworth, WJ and Rarity, JG and Tame, MS},
  journal={Nature communications},
  volume={5},
  number={1},
  pages={1--12},
  year={2014},
  publisher={Nature Publishing Group}
}

@inproceedings{mhalla2011,
  title={Which graph states are useful for quantum information processing?},
  author={Mhalla, M. and Murao, M. and Perdrix, S. and Someya, M. and Turner, PS},
  booktitle={Conference on Quantum Computation, Communication, and Cryptography},
  pages={174--187},
  year={2011},
  organization={Springer}
}

@article{oeding2011,
  title={Set-theoretic defining equations of the variety of principal minors of symmetric matrices},
  author={Oeding, L.},
  journal={Algebra \& Number Theory},
  volume={5},
  number={1},
  pages={75--109},
  year={2011},
  publisher={Mathematical Sciences Publishers}
}

@article{oeding2017,
  title={The quadrifocal variety},
  author={Oeding, L.},
  journal={Linear Algebra and Its Applications},
  volume={512},
  pages={306--330},
  year={2017},
  publisher={Elsevier}
}

@article{griffin2006,
  title={Principal minors, Part II: The principal minor assignment problem},
  author={Griffin, K. and Tsatsomeros, MJ},
  journal={Linear Algebra and its applications},
  volume={419},
  number={1},
  pages={125--171},
  year={2006},
  publisher={Elsevier}
}

@article{kenyon2014,
  title={Principal minors and rhombus tilings},
  author={Kenyon, R. and Pemantle, R.},
  journal={Journal of Physics A: Mathematical and Theoretical},
  volume={47},
  number={47},
  pages={474010},
  year={2014},
  publisher={IOP Publishing}
}

@article{cabello2009,
  title={Entanglement in eight-qubit graph states},
  author={Cabello, A. and L\'{o}pez-Tarrida, A.J and Moreno, P. and Portillo, J.R},
  journal={Physics Letters A},
  volume={373},
  number={26},
  pages={2219--2225},
  year={2009},
  publisher={Elsevier}
}

@article{cabello2011,
  title={Optimal preparation of graph states},
  author={Cabello, A. and Danielsen, L.E and L\'{o}pez-Tarrida, A.J and Portillo, J.R},
  journal={Physical Review A},
  volume={83},
  number={4},
  pages={042314},
  year={2011},
  publisher={APS}
}

@article{holtz2007,
  title={Hyperdeterminantal relations among symmetric principal minors},
  author={Holtz, O. and Sturmfels, B.},
  journal={Journal of Algebra},
  volume={316},
  number={2},
  pages={634--648},
  year={2007},
  publisher={Elsevier}
}

@article{van2005,
  title={Local unitary versus local Clifford equivalence of stabilizer states},
  author={Van den Nest, M. and Dehaene, J. and De Moor, B.},
  journal={Physical Review A},
  volume={71},
  number={6},
  pages={062323},
  year={2005},
  publisher={APS}
}

@article{danielsen2005,
	title={On self-dual quantum codes, graphs, and Boolean functions},
	author={Danielsen, L.E.},
	journal={arXiv preprint quant-ph/0503236},
	year={2005}
}

@article{holweck2014,
	title={A Notable Relation between N-Qubit and $2^{N-1}$-Qubit Pauli Groups via Binary LGr(N, 2N)},
	author={Holweck, F. and Saniga, M. and L\'{e}vay, P.},
	journal={SIGMA. Symmetry, Integrability and Geometry: Methods and Applications},
	volume={10},
	pages={041},
	year={2014},
	publisher={SIGMA. Symmetry, Integrability and Geometry: Methods and Applications}
}

@article{levay2013,
  title={Grassmannian connection between three-and four-qubit observables, Mermin’s contextuality and black holes},
  author={L\'{e}vay, P. and Planat, M. and Saniga, M.},
  journal={Journal of High Energy Physics},
  volume={2013},
  number={9},
  pages={37},
  year={2013},
  publisher={Springer}
}

@article{van2019,
  title={Lagrangian Grassmannians and spinor varieties in characteristic two},
  author={Van Geemen, B. and Marrani, A.},
  journal={SIGMA. Symmetry, Integrability and Geometry: Methods and Applications},
  volume={15},
  pages={064},
  year={2019},
  publisher={SIGMA. Symmetry, Integrability and Geometry: Methods and Applications}
}

@article{oeding2009,
  title={G-varieties and the principal minors of symmetric matrices},
  author={Oeding, L.},
  journal={PhD Thesis},
  year={2009},
  publisher={Texas A\&M University}
}

@article{levay2017,
  title={Magic three-qubit Veldkamp line: A finite geometric underpinning for form theories of gravity and black hole entropy},
  author={L\'{e}vay, P. and Holweck, F. and Saniga, M.},
  journal={Physical Review D},
  volume={96},
  number={2},
  pages={026018},
  year={2017},
  publisher={APS}
}

@article{havlicek2009,
  title={Factor-group-generated polar spaces and (multi-) qudits},
  author={Havlicek, H. and Odehnal, B. and Saniga, M.},
  journal={SIGMA. Symmetry, Integrability and Geometry: Methods and Applications},
  volume={5},
  pages={096},
  year={2009},
  publisher={SIGMA. Symmetry, Integrability and Geometry: Methods and Applications}
}

@book{harris2013,
  title={Algebraic geometry: a first course},
  author={Harris, J.},
  volume={133},
  year={2013},
  publisher={Springer Science \& Business Media}
}

@article{landsberg2012,
  title={Tensors: geometry and applications},
  author={Landsberg, J.M},
  journal={Representation theory},
  volume={381},
  number={402},
  pages={3},
  year={2012}
}

\end{document}